\documentclass{article}

\usepackage{hyperref}
\usepackage{graphicx} 
\usepackage{blindtext}
\usepackage{titlesec}
\usepackage{mathrsfs}
\usepackage{physics}
\usepackage{tensor}
\usepackage{amsfonts}
\usepackage{mathtools}
\usepackage[a4paper, total={17.6cm, 26.13cm}]{geometry}
\usepackage{url}
\usepackage{amsmath}
\usepackage{amsthm}
\usepackage{scalerel}
\usepackage{subcaption}
\usepackage{xcolor}

\newtheorem{definition}{Definition}[section]
\newtheorem{theorem}{Theorem}[section]
\newtheorem{corollary}{Corollary }[theorem]
\newtheorem{lemma}[theorem]{Lemma}

\DeclareSymbolFont{toneletters}{T1}{\familydefault}{m}{it}
\SetSymbolFont{toneletters}{bold}{T1}{\familydefault}{bx}{it}

\DeclareMathSymbol\ethm{\mathord}{toneletters}{"F0}

\def\be{\begin{equation}}
\def\ee{\end{equation}}
\def\ba{\begin{eqnarray}}
\def\ea{\end{eqnarray}}

\DeclareMathOperator{\ethmmbar}{\ethm_{c}^{\prime}}
\DeclareMathOperator{\ethmm}{\ethm_{c}}

\begin{document}
\title{On the third law of black hole mechanics for supersymmetric black holes}
\author{Aidan~M.~McSharry and Harvey~S.~Reall\\ {\footnotesize Department of Applied Mathematics and Theoretical Physics, University of Cambridge} \\ {\footnotesize Wilberforce Road, Cambridge CB3 0WA, United Kingdom}\\ {\footnotesize amm323@cam.ac.uk, hsr1000@cam.ac.uk}}
\maketitle

\begin{abstract}
Recently it has been shown that the third law of black hole mechanics can be violated: an exactly extremal Reissner-Nordstr\"om black hole can form in finite time in gravitational collapse of matter with a large charge to mass ratio. However, it has also been proved that this cannot happen if the matter satisfies a ``supersymmetric'' lower bound on its energy in terms of its charge. This paper proves an analogous result for black holes with a negative cosmological constant. The result states that a supersymmetric Kerr-Newman-anti de Sitter black hole cannot form in gravitational collapse of charged matter satisfying the supersymmetric bound. The results for zero or negative cosmological constant are extended to apply to two-sided black holes: it is proved that an initially non-extremal black hole cannot evolve to a supersymmetric black hole in finite time, irrespective of whether or not the initial black hole was formed in gravitational collapse. 
\end{abstract}

\section{Introduction}

The third law of black hole mechanics states that it is impossible for a non-extremal black hole to become extremal in finite time in classical General Relativity. This was initially conjectured in \cite{bardeen1973four} and a proof was presented in \cite{israel1986third}. However, counterexamples to this law were found recently: Kehle and Unger constructed third-law violating solutions of Einstein-Maxwell theory coupled to a massless charged scalar field \cite{kehle2025gravitational} (or charged Vlasov matter  with small mass parameter \cite{kehle2024extremalblackholeformation}). These solutions describe spherically symmetric gravitational collapse of the scalar field to form an exactly extremal Reissner-Nordstr\"om (RN) black hole in finite time, with an intermediate phase in which the solution is exactly Schwarzschild at the horizon.

These examples involve matter with a large charge to mass ratio. One expects it to be harder to form an extremal RN black hole if the charge to mass ratio is bounded. In particular, consider matter satisfying the ``local mass-charge inequality'' \cite{gibbons1982bogomolny,ghhp}:
\begin{equation}
\label{eq:localbps}
 T^{(m)}_{00} \ge \sqrt{T^{(m)}_{0i}T^{(m)}_{0i} + J_0^2}
\end{equation}
where indices $\{0,i\}$ $(i=1,2,3)$ refer to an arbitrary orthonormal frame,  $T^{(m)}_{ab}$ is the energy-momentum tensor of the charged matter and  $J_a$ its electric current density. This inequality is a strengthened version of the dominant energy condition. It has been proved recently that if matter respects this inequality then it is impossible for an extremal RN black hole to form in finite time in gravitational collapse \cite{reall2025lawblackholemechanics}. This excludes spacetimes of the type constructed by Kehle and Unger, which are exactly extremal RN everywhere outside the horizon at late time. It also excludes the much larger class of (not necessarily spherically symmetric) spacetimes describing gravitational collapse to form a black hole that is extremal RN in a neighbourhood of the horizon at late time but differs from extremal RN outside this neighbourhood because of the presence of outgoing electromagnetic and/or gravitational radiation and/or outgoing matter. A Penrose diagram for such a spacetime is shown in Figure \ref{fig:penrose_collapse}(a). 

The proof of this result combines ideas from two sources. First, the Gibbons-Hull spinorial proof of the global mass-charge inequality (``BPS bound'') for Einstein-Maxwell theory coupled to matter satisfying \eqref{eq:localbps} \cite{gibbons1982bogomolny}. Second, the spinorial quasilocal mass construction of Dougan and Mason \cite{douganquasi}. Using these ideas it was shown that if $\Sigma$ is a spacelike 3-surface with boundary $S$ and the (spacetime) metric, extrinsic curvature, and Maxwell field on $S$ coincide with those on a horizon cross-section of extremal RN then $\Sigma$ cannot be compact. Taking $\Sigma$ to be a surface inside an extremal RN black hole formed in gravitational collapse, with $S$ a late-time horizon cross-section (as shown in Fig. \ref{fig:penrose_collapse}(a)) immediately gives a contradiction since such $\Sigma$ would be compact.

\begin{figure*}[t!]
    \centering
    \begin{subfigure}[t]{0.5\textwidth}
        \centering
        \includegraphics[width =6.5cm, height = 5.5cm]{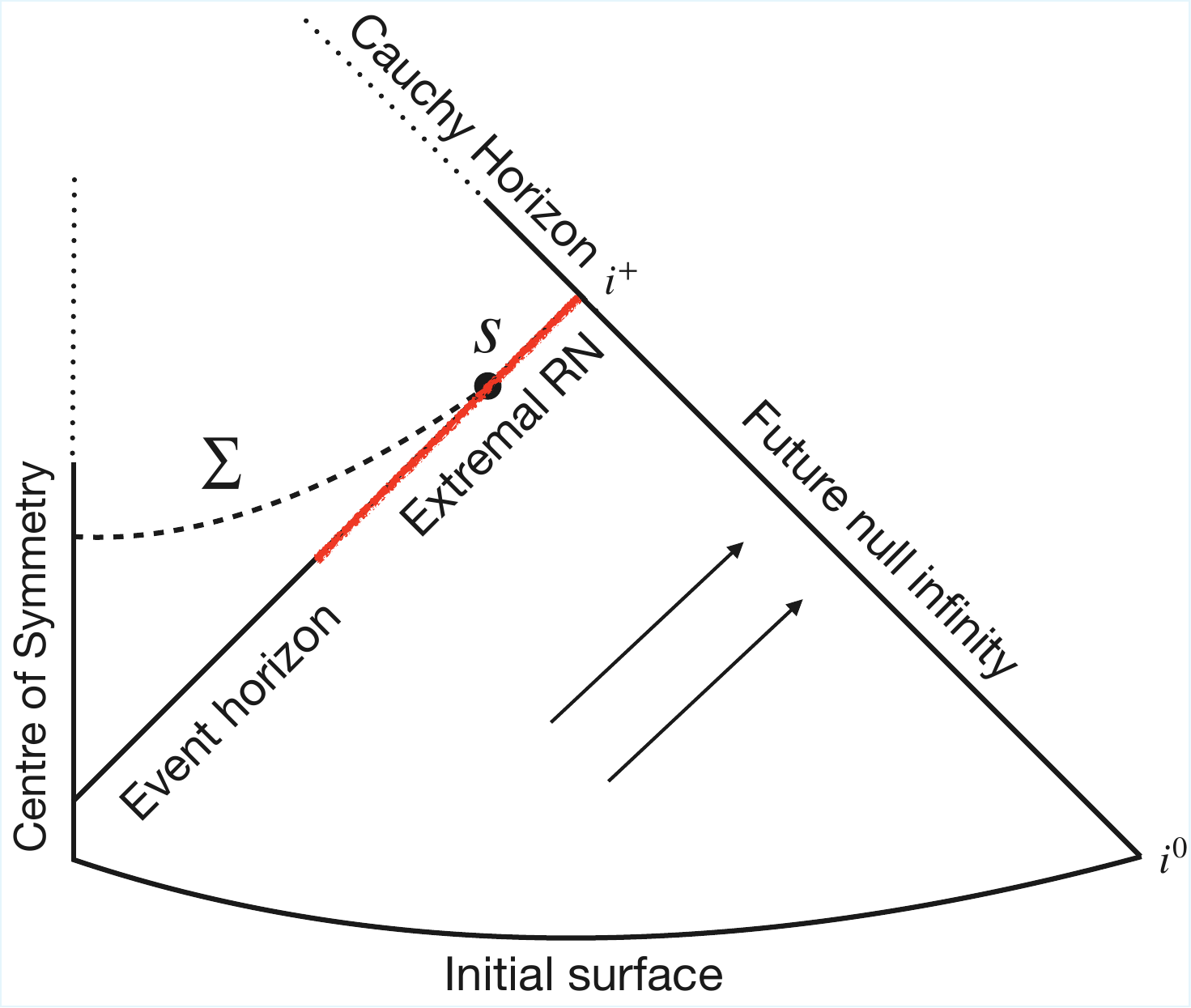}
        \caption{}
    \end{subfigure}%
    ~ 
    \begin{subfigure}[t]{0.5\textwidth}
        \centering
        \includegraphics[width = 5.0cm, height = 6.0cm]{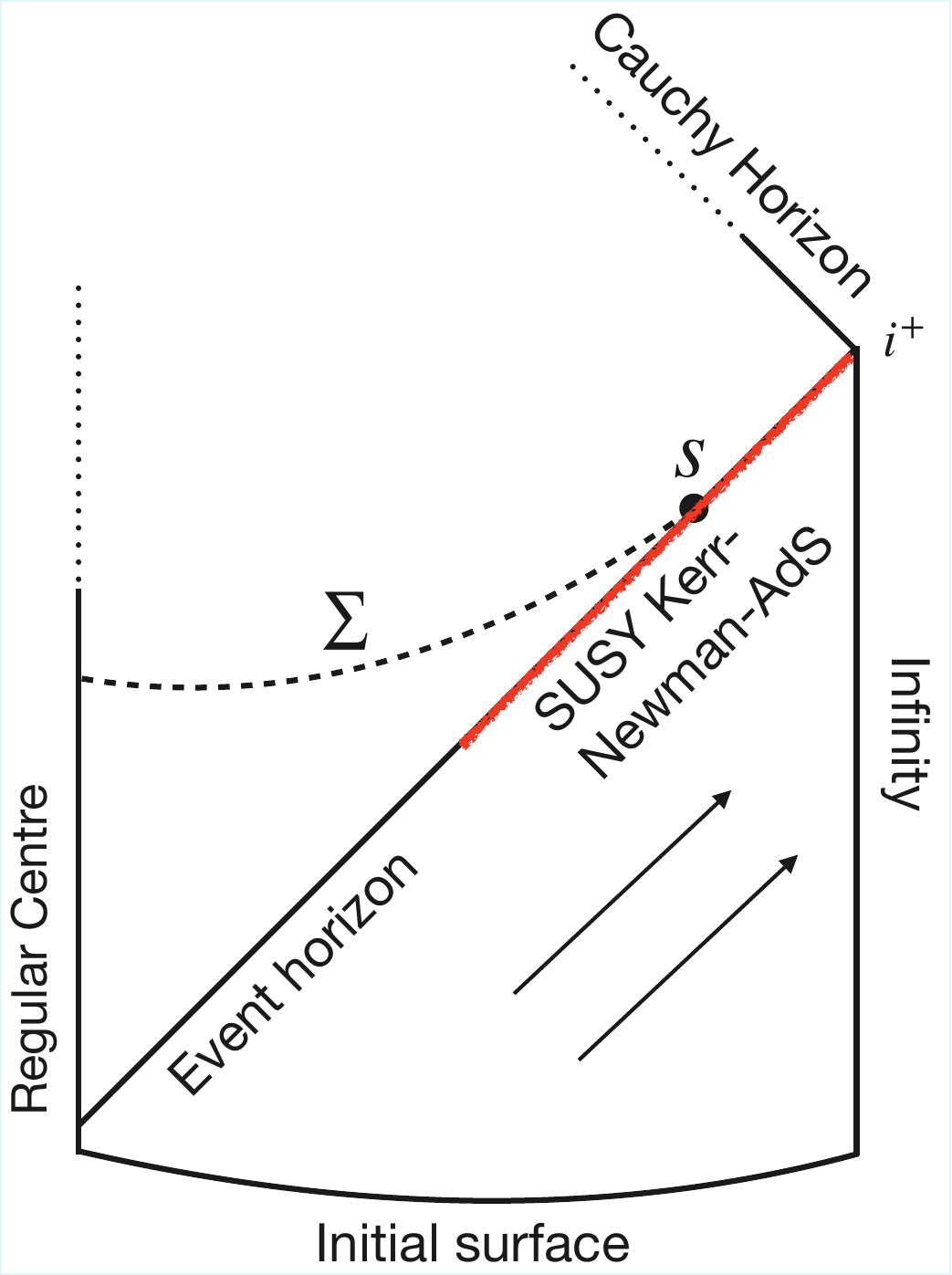}
        \caption{}
    \end{subfigure}
    \caption{(a) Penrose diagram showing the formation of an extremal Reissner-Nordstr\"om black hole in spherically symmetric gravitational collapse of charged matter. The final black hole coincides with extremal Reissner-Nordstr\"om in the red-shaded region. The arrows denote outgoing matter or electromagnetic/gravitational radiation (the latter would break spherical symmetry). Ref \cite{reall2025lawblackholemechanics} proves that such spacetimes do not exist if the matter satisfies the local mass-charge inequality \eqref{eq:localbps}. (b) Diagram showing the formation of a supersymmetric Kerr-Newman-adS black hole in gravitational collapse. Such a process could not be spherically symmetric so this diagram is only schematic. We prove that spacetimes like (b) cannot exist if the charged matter satisfies the local mass-charge inequality, independently of any boundary conditions at infinity.}
\label{fig:penrose_collapse}
\end{figure*}
In the present paper we will extend this result in two ways. First, we will allow for a negative cosmological constant, i.e., we will consider asymptotically (locally) anti-de Sitter (adS) black holes. We will prove that a {\it supersymmetric} adS black hole cannot form in finite time in gravitational collapse. Here, ``supersymmetric'' refers to the existence of a non-trivial ``supercovariantly constant'' spinor. Such black holes are necessarily extremal. Extremal Reissner-Nordstr\"om is supersymmetric but extremal Reissner-Nordstr\"om-adS is not. However, there exists a 1-parameter subfamily of Kerr-Newman-adS black holes that are supersymmetric. These are the solutions with $E=|J|/\ell + |Q|$ and $P=0$ where $E$ is the energy, $J$ the angular momentum, $\ell$ the adS radius, $Q$ the electric charge and $P$ the magnetic charge \cite{Kosteleck__1996,caldarelli1999supersymmetry}. This family can be parameterized with $Q$ and it has $J \ne 0$. More precisely, we will prove that a black hole formed in finite time in gravitational collapse cannot admit a horizon cross-section on which the (spacetime) metric, extrinsic curvature and Maxwell field are the same as those of a supersymmetric adS black hole:

\begin{theorem} 
\label{thm:inf1}
(Informal version.)
 Consider a smooth spacetime satisfying the Einstein-Maxwell equations with a negative cosmological constant and charged matter satisfying the local mass-charge inequality \eqref{eq:localbps}. Let $\Sigma$ be a  smooth, compact, connected, spacelike $3$-surface with $\partial \Sigma=S$ where $S$ is a smooth, compact, connected $2$-surface. Then there does not exist a local diffeomorphism that maps the spacetime metric, Maxwell field and extrinsic curvature on $S$ to the corresponding quantities on a horizon cross-section of a supersymmetric black hole.
\end{theorem}
A precise version of this result is stated as Corollary \ref{theorem:third law in generality} below. The result is of a quasi-local nature, so it does not depend on any choice of boundary conditions at infinity in adS e.g. it holds for both ``reflecting'' and ``transparent'' boundary conditions. Fig. \ref{fig:penrose_collapse}(b) shows an example of a spacetime whose existence is excluded by this result. 

The result of \cite{reall2025lawblackholemechanics} and our new result just described show that if matter satisfies \eqref{eq:localbps} then a supersymmetric black hole cannot form in finite time in gravitational collapse. This implies a version of the third law of black hole mechanics, namely that ``a non-extremal black hole cannot become supersymmetric in finite time'' provided that the initial non-extremal black hole was formed in gravitational collapse. However, it does not cover situations where the initial black hole is not formed in collapse, for example it does not exclude the possibility of forming a supersymmetric black hole starting from a non-extremal black hole with two asymptotically flat (or adS) regions. This possibility is of particular interest in the presence of a negative cosmological constant because there exist supersymmetric black hole solutions which do not belong to the Kerr-Newman-adS family. These solutions are static, with horizon cross-sections given by Riemann surfaces of genus greater than one \cite{caldarelli1999supersymmetry}. They carry magnetic charge and therefore cannot form in gravitational collapse of electrically charged matter.\footnote{
For $\Lambda=0$, the result of \cite{reall2025lawblackholemechanics} allows for magnetically charged matter but for $\Lambda<0$ we cannot include such matter because we require the existence of an electromagnetic potential to define the supercovariant derivative.} Our theorem above has nothing to say about such black holes. But one can ask whether it is possible for such a black hole to form in finite time starting from one or more non-extremal magnetically charged black holes.

The second main achievement of this paper is to extend the results of \cite{reall2025lawblackholemechanics} and the theorem above to prove a third law for supersymmetric black holes that are not necessarily formed in gravitational collapse, as in the case just discussed. Our result says that an initially non-extremal black hole cannot become supersymmetric in finite time. Here, ``non-extremal'' is defined in the same way as in \cite{israel1986third}, namely the existence of a trapped surface. More precisely, consider a spacetime containing a spacelike 3-surface $\Sigma$ with boundary $\partial \Sigma = S \cup T$ where $T$ is a trapped surface (possibly disconnected, inside the ``initial'' black hole(s)) and $S$ is a cross-section of the event-horizon at late time (the ``final'' black hole). Assume that
the (spacetime) metric, extrinsic curvature, and Maxwell field on $S$ coincide with those on a horizon cross-section of a supersymmetric black hole. See Fig. \ref{fig:inner_bdy}. We prove that such $\Sigma$ cannot exist if matter respects the local mass-charge inequality \eqref{eq:localbps} on $\Sigma$:

\begin{theorem} 
\label{thm:inf2}
(Informal version.)
Consider a smooth spacetime satisfying the Einstein-Maxwell equations with vanishing or negative cosmological constant and charged matter satisfying the local mass-charge inequality \eqref{eq:localbps}. Let $\Sigma$ be a  smooth, compact, connected, spacelike $3$-surface with $\partial \Sigma=S \cup T$ where $S$, the ``outer'' boundary of $\Sigma$, is a smooth, compact, connected, surface and where $T$, the ``inner'' boundary of $\Sigma$, is an outer trapped 2-surface. Then there does not exist a local diffeomorphism that maps the spacetime metric, Maxwell field and extrinsic curvature of $S$ to the corresponding quantities on a horizon cross-section of a supersymmetric black hole.
\end{theorem}
A precise version of this result is stated as Corollary \ref{cor:3rd law with inner boundary} below. 

The main ingredient used to obtain all of our results is the idea of a ``supersymmetric surface'' \cite{reall2025lawblackholemechanics}. This is a spacelike 2-surface on which there exists a non-trivial spinor satisfying the tangential components of the supercovariantly constant condition. In a non-supersymmetric spacetime, if there exists a local diffeomorphism mapping the (spacetime) metric, extrinsic curvature and Maxwell field on a 2-surface $S$ to the corresponding quantities on a 2-surface in a supersymmetric spacetime, then $S$ is a supersymmetric surface. Ideas similar to the quasilocal mass construction of \cite{douganquasi} can be used to prove a rigidity property for such surfaces in spacetimes satisfying \eqref{eq:localbps}: if $\Sigma$ is a 3-surface with $\partial \Sigma = S$ (or $\partial \Sigma = S\cup T$ with $T$ weakly outer trapped) then there exists a supercovariantly constant spinor on $\Sigma$, which strongly constrains the form of the spacetime on $\Sigma$. See Theorems \ref{thm:equiv to 2.6} and \ref{thm:equiv to 2.6 but with inner boundary} below. If the local diffeomorphism maps $S$ to a horizon cross-section of a supersymmetric black hole, then $S$ has the additional property of being marginally trapped, which leads to even stronger rigidity statements given in Theorems \ref{thm:equiv to 2.8} and \ref{thm:equiv to 2.8 but with inner boundary}, with Theorems \ref{thm:inf1} and \ref{thm:inf2} arising as corollaries of these theorems.

This paper is organized as follows. In section \ref{sec:prelim} we introduce our conventions, notation and some background material. In section \ref{sec:main} we start by introducing the idea of a supersymmetric surface $S$ and then study the situation where there exists a compact spacelike surface $\Sigma$ with $\partial \Sigma = S$. This section establishes the rigidity results of Theorems \ref{thm:equiv to 2.6} and \ref{thm:equiv to 2.8}, with Theorem \ref{thm:inf1} above arising as a corollary. In section \ref{sec:trapped surfaces} we consider the generalisation where $\partial \Sigma = S\cup T$ with $T$ a weakly outer trapped surface, establishing Theorems \ref{thm:equiv to 2.6 but with inner boundary} and \ref{thm:equiv to 2.8 but with inner boundary}, with Theorem \ref{thm:inf2} above arising as a corollary. There are $3$ Appendices covering technical material used in some of the proofs.

\begin{figure*}[t!]
    \centering
    \begin{subfigure}[t]{0.5\textwidth}
        \centering
        \includegraphics[width=8cm,height = 5.25cm]{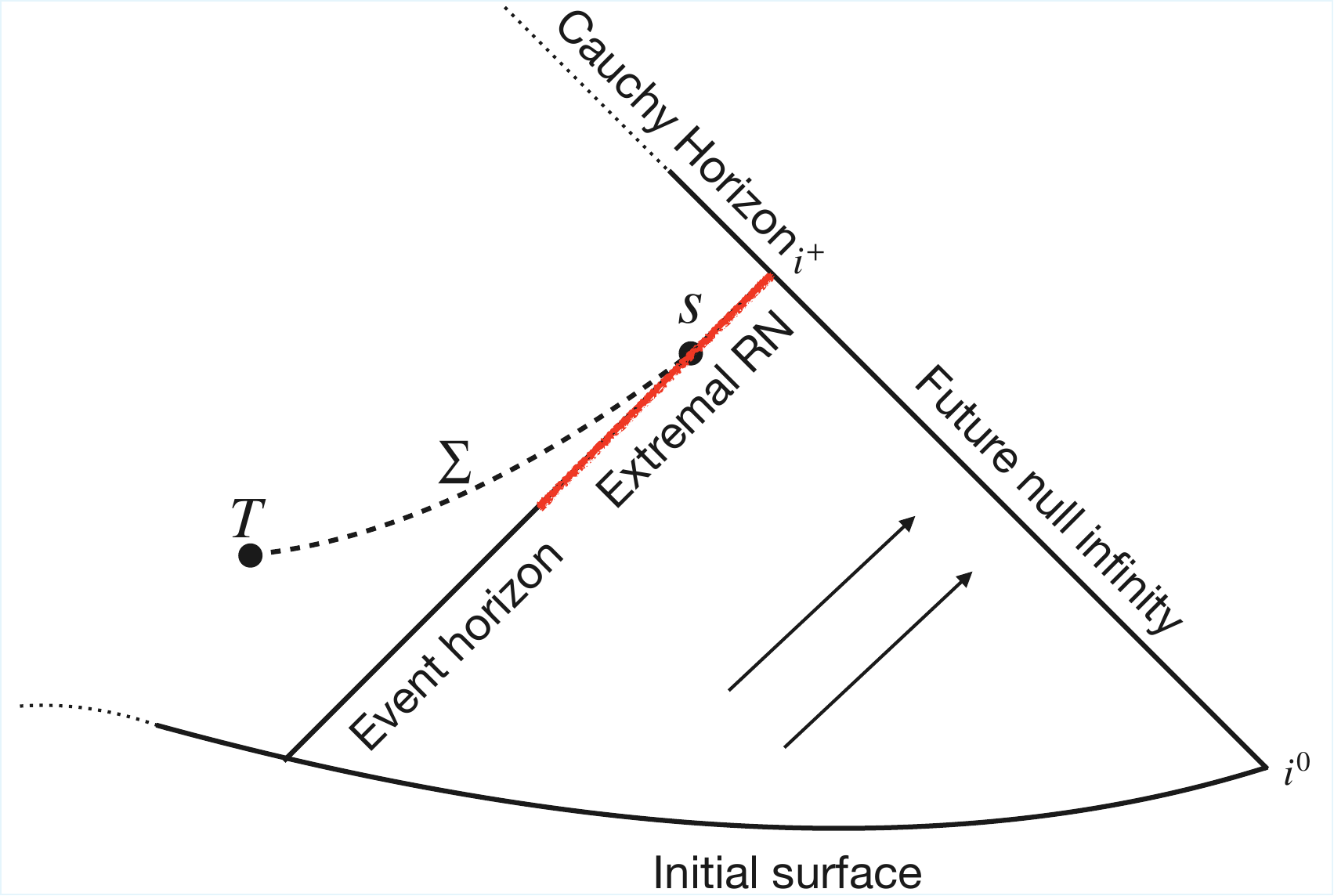}
        \caption{}
    \end{subfigure}%
    ~
    \begin{subfigure}[t]{0.5\textwidth}
        \centering
        \includegraphics[width=5.6cm, height = 5.5cm]{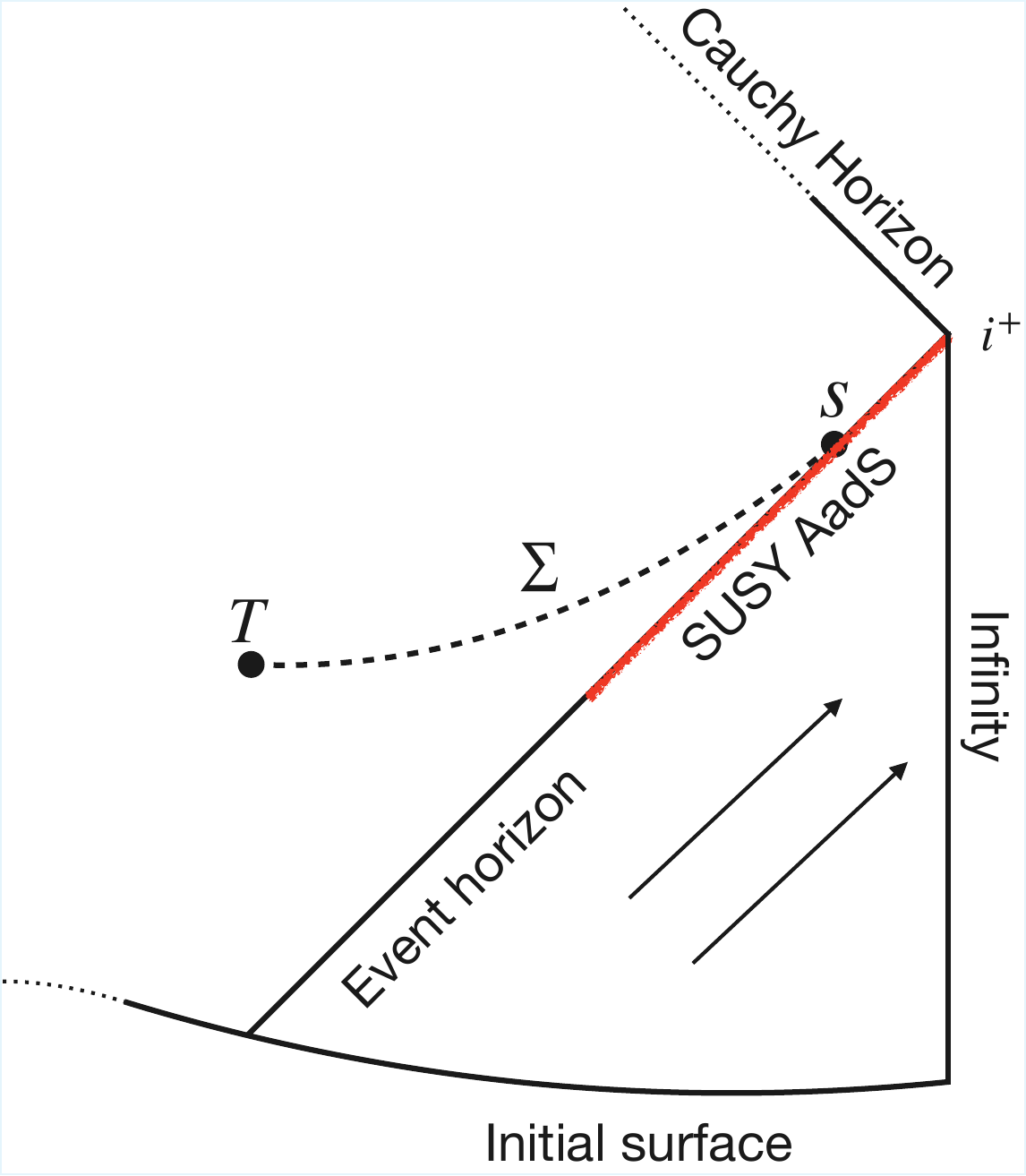}
        \caption{}
    \end{subfigure}
    \caption{(a) Penrose diagram showing a third law violating spacetime in which an initially non-extremal black hole becomes extremal but without assuming that the initial black hole was formed in gravitational collapse. The initial surface might extend into another asymptotically flat region to the left of the figure. The presence of a trapped surface $T$ shows that the initial black hole is non-extremal. The final black hole coincides with extremal Reissner-Nordstr\"om in the red-shaded region. The spacelike surface $\Sigma$ extends from $T$ to a Reissner-Nordstr\"om horizon cross-section $S$. We prove that such a surface cannot exist, and hence such a spacetime cannot exist, if matter satisfies the local mass-charge inequality. (b) The corresponding diagram for negative cosmological constant, where the final black hole is a supersymmetric asymptotically adS black hole. 
    }
\label{fig:inner_bdy}
\end{figure*}

\section{Preliminaries}
\label{sec:prelim}

\subsection{Conventions}

As in \cite{reall2025lawblackholemechanics}, we shall use the Newman-Penrose formalism \cite{penrose1984spinorsvol1} and so it is convenient to work with a negative signature metric. Gamma matrices satisfy $\{\gamma^a,\gamma^b\} = 2g^{ab}$  where $a,b \in \{0,1,2,3\}$ refer to an orthonormal basis. We take $\gamma^0$ hermitian and $\gamma^i$ ($i=1,2,3$) anti-hermitian.

\subsection{Equations of motion}\label{sec:setup}

We will consider Einstein-Maxwell theory coupled to charged matter in a four-dimensional spacetime in the presence of a negative cosmological constant:
\be
\Lambda = -6 K^2
\ee
with $K>0$. The matter has energy-momentum tensor $T^{(m)}_{ab}$, and electric current density $J_a$. The matter is assumed to satisfy \eqref{eq:localbps} which can also be written as \cite{reall2025lawblackholemechanics}
\begin{equation} \label{eq:local mass-charge inequality}
    T^{(m)}_{ab}V^aV^b \geq 0, \qquad T^{(m)}_{ab}V^b T^{(m)}{^a}{_c}V^c \geq (J_aV^a)^2,
\end{equation}
where $V$ is any causal vector. The Einstein equation is
\begin{equation} \label{eq:efe}
    R_{ab} - \frac{1}{2}Rg_{ab} - \Lambda g_{ab} = 8\pi\big(T^{MW}_{ab} + T_{ab}^{(m)}\big),
\end{equation}
where
\begin{equation}
    T^{MW}_{ab} = -\frac{1}{4\pi}\Big( F_a{^c}F_{bc} - \frac{1}{4}F_{cd}F^{cd} g_{ab} \Big).
\end{equation}
Maxwell's equations are:
\begin{equation}\label{eq:maxwell}
    d\star F = 4\pi \star J,\qquad dF = 0
\end{equation}
The final equation implies that, locally, we can introduce a gauge potential $A_a$ such that $F=dA$.

\subsection{Newman-Penrose formalism}

\label{sec:NP}

We denote a spin frame as $\{o^A, \iota^A\}$, where $o_A \iota^A = 1$. A Weyl spinor can be written in this frame as $\lambda^A = \lambda^0 o^A + \lambda^1 \iota^A$.
This gives $\lambda_A = -\lambda_1o_A + \lambda_0 \iota_A$, where $\lambda_0 = \lambda^1$ and $\lambda_1 = -\lambda^0$. The complex conjugate spin frame is $\{\bar{o}^{A{'}},\bar{\iota}^{A{'}}\}$.


A spin frame induces a Newman-Penrose (NP) tetrad \cite{penrose1984spinorsvol1}, $\{l^a,n^a,m^a,\bar{m}^a\}$ according to
\begin{equation}
    l^a = o^A \bar{o}^{A{'}}, \qquad n^a = \iota^A \bar{\iota}^{A{'}}, \qquad m^a = o^A \bar{\iota}^{A{'}}, \qquad \bar{m}^{a}=\iota^A \bar{o}^{A{'}},
\end{equation}
where $l^a$ and $n^a$ are real null vectors, while $m^a$ is a complex null vector whose real and imaginary parts are spacelike and orthogonal. The connection components in such a tetrad are described by the NP scalars.

From such a NP tetrad, one can construct an orthonormal basis:
\begin{equation}
\begin{split}
        e_{0}{^a} &= \frac{1}{\sqrt{2}}(l^a + n^a), \text{  } e_{3}{^a} = \frac{1}{\sqrt{2}}(l^a-n^a),\\
        e_{1}{^a} &= \frac{1}{\sqrt{2}}(m^a + \bar{m}^a), \text{    } e_{2}{^a} = \frac{1}{\sqrt{2}}i(\bar{m}^a - m^a).
\end{split}
\end{equation}
In what follows, unless stated otherwise, the indices $0,1,2,3$ will refer to this orthonormal basis.

We will often consider a compact spacelike 3-surface $\Sigma$ with boundary $\partial \Sigma = S \cup T$ where $S$ is the ``outer'' boundary of $\Sigma$ (which we assume to be connected) and $T$ is the ``inner'' boundary of $\Sigma$ (which might be empty or might have multiple connected components). On a spacelike 2-surface, $S$ or $T$, there are precisely two future-directed null directions normal to the surface. We choose our NP tetrad so that $l^a$ and $n^a$ are normal to $S$ and $T$, with $l^a$ pointing outwards (i.e. out of $\Sigma$ on $S$, into $\Sigma$ on $T$) and $n^a$ pointing inwards. 

In the compact notation of \cite{geroch1973space}, the NP variables of interest to us are $\rho,\sigma,\rho{'},\sigma{'}$. On $S$ or $T$, $\rho$ and $\rho'$ are real. $\rho, \sigma$ describe the expansion and shear of outgoing null geodesics and $\rho',\sigma'$ describe the expansion and shear of ingoing null geodesics. The signs are such that positive $\rho$ or $\rho'$ corresponds to converging geodesics. We say that $S$ or (a connected component of) $T$ is {\it trapped} if it has $\rho'>0$ and $\rho>0$, {\it outer trapped} if it has $\rho>0$, {\it weakly outer trapped} if it has $\rho \ge 0$ and {\it marginally outer trapped} if it has $\rho \equiv 0$. 

\subsection{Supercovariant derivative on spinors}\label{sec:derivatives of spinors}

We will be using spinorial methods motivated by supersymmetry. In 4d gauged supergravity the supersymmetry transformation parameter is a spinor field $\epsilon$ charged under the $U(1)$ gauge field \cite{FREEDMAN1977221}. The supercovariant derivative acting on such a spinor is defined as \cite{Romans_1992}

\begin{equation}\label{eq:supercovariant derivative on dirac spinor}
    \hat{\nabla}_a \epsilon \coloneqq D_a \epsilon + i\frac{K}{\sqrt{2}}\gamma_a \epsilon + \frac{1}{4} F_{bc}\gamma^b\gamma^c\gamma_a \epsilon,
\end{equation}
where $D_a \coloneqq \nabla_a -i\sqrt{2}K A_a$ is a $U(1)$-gauge covariant derivative.
A Dirac spinor decomposes into a pair of Weyl spinors: $\epsilon = (\lambda^A,\bar{\mu}_{A{'}})$ (note the difference in the placement of the overbar compared with \cite{reall2025lawblackholemechanics}).
The spinor supercovariant derivative can be written in terms of these Weyl spinors as
\begin{equation} \label{eq:supercovariant derivative on weyl spinors}
    \hat{\nabla}_{AA{'}}\lambda_{B} = D_{AA{'}}\lambda_{B} + (K\epsilon_{AB} + \sqrt{2}\phi_{AB}) \bar{\mu}_{A{'}}  \qquad\text{ and  } \qquad\hat{\nabla}_{AA{'}}\bar{\mu}_{B^{'}} =D_{AA{'}}\bar{\mu}_{B^{'}} + (K\epsilon_{A{'}B^{'}} - \sqrt{2}\bar{\phi}_{A{'}B^{'}}) \lambda_A,
\end{equation}
where $\phi_{AB}$ is the symmetric spinor describing the Maxwell field \cite{penrose1984spinorsvol1}:
\begin{equation}
    F_{ab} = \phi_{AB}\epsilon_{A{'}B^{'}} + c.c.
\end{equation}

Since our spinor is charged w.r.t. the $U(1)$ gauge field, it is convenient to define gauge-covariant modifications $\ethmm$ and $\ethmmbar$ of the GHP \cite{geroch1973space} derivative operators $\ethm$ and $\bar{\ethm}$ (the subscript ``c'' refers to ``charged''). These are defined by 
\begin{equation}
\label{eq:ethmc}
    \ethmm \coloneqq \ethm - i\sqrt{2}K (m\cdot A) \text{    and    } \ethmmbar \coloneqq \bar{\ethm} - i\sqrt{2}K (\bar{m}\cdot A).
\end{equation}
The complex conjugates of these operators are
\begin{equation}
    \bar{\ethmm} \coloneqq \bar{\ethm} + i\sqrt{2}K (\bar{m}\cdot A) \text{    and    } \bar{\ethmmbar} \coloneqq \ethm + i\sqrt{2}K (m\cdot A).
\end{equation}
These operators satisfy the following equation on any closed, spacelike 2-surface $S$ (with $m^a$, $\bar{m}^a$ tangent to $S$)
\begin{equation}\label{eq:integration by parts for gauge covariant eth}
    \int_S \lambda_1 (\ethmm \bar{\lambda}_{0{'}}) = - \int_S \bar{\lambda}_{0{'}} \bar{\ethmmbar}\lambda_1 ,
\end{equation}
and likewise for $\ethmmbar$.

\subsection{Vector and scalar fields}

From a Dirac spinor $\epsilon = (\lambda^A,\bar{\mu}_{A{'}})$  we can construct a vector and a complex scalar \cite{tod1983all}. In terms of 2-spinors:
\begin{equation} \label{eq:vector formed from spinors}
    X^a = \frac{1}{\sqrt{2}}(\bar{\lambda}^{A{'}}\lambda^{A} + \bar{\mu}^{A{'}}\mu^{A}) \qquad V = \lambda_A\mu^A.
\end{equation}
which satisfy
\begin{equation}
    X_aX^a = V\bar{V}
\end{equation}
and so $X^a$ is causal (or zero). $V$ vanishes if, and only if, $\lambda_A$ and $\mu_A$ are linearly dependent. $X^a$ vanishes if, and only if, $\lambda_A$ and $\mu_A$ vanish.

\section{Third law for supersymmetric black holes formed in gravitational collapse} \label{sec:main}

\subsection{Supersymmetric surfaces}

Following \cite{reall2025lawblackholemechanics} it is convenient to define
\begin{definition} [Supersymmetric Surface]
    Let $S$ be a smooth, connected, spacelike 2-surface in a smooth spacetime containing a Maxwell field. $S$ is a {\rm supersymmetric surface} if there exists a not identically zero Dirac spinor field, $\epsilon$, defined on $S$ such that, for any vector field $t^a$ tangent to $S$, $t^a\hat{\nabla}_a \epsilon = 0$.
\end{definition}
A standard argument shows that such a spinor field must be non-vanishing everywhere on $S$  \cite{reall2025lawblackholemechanics}. Note that this definition depends on the cosmological constant through the definition of the supercovariant derivative.

A {\it spacetime} is said to be supersymmetric if it admits a globally defined supercovariantly constant spinor. Any smooth, connected, spacelike 2-surface in such a spacetime is a supersymmetric surface. However, we are interested in the possible existence of supersymmetric surfaces in non-supersymmetric spacetimes. We want to consider spacetimes describing gravitational collapse to form a black hole that is ``supersymmetric on the horizon'' at late time. A late-time horizon cross-section of such a black hole would be an example of a supersymmetric surface in a non-supersymmetric spacetime.

Given a supersymmetric surface $S$, if we pick a NP basis adapted to $S$ (i.e. so that $m^a$ and $\bar{m}^a$ are tangential to $S$) then
the condition of supersymmetry of $S$ is $ \bar{m}^a\hat{\nabla}_a \epsilon=m^a \hat{\nabla}_a \epsilon  = 0$. In terms of Weyl spinors, using the modified GHP derivative operators \eqref{eq:ethmc}, the first of these equations is

\begin{subequations}
    \begin{equation}\label{eq:m susy condition a}
        \ethmmbar \lambda_{1} + \sigma{'}\lambda_{0} + \sqrt{2}\phi_{11} \bar{\mu}_{1{'}} = 0, \qquad \ethmmbar \bar{\mu}_{1^{'}} + \rho{'} \bar{\mu}_{0^{'}} + (-K - \sqrt{2}\bar{\phi}_{0^{'}1^{'}})\lambda_1 =0,
    \end{equation}
    \begin{equation}\label{eq:m susy condition b}
        \ethmmbar \lambda_{0} + \rho \lambda_{1} + (K + \sqrt{2}\phi_{01})\bar{\mu}_{0{'}}= 0, \qquad \ethmmbar\bar{\mu}_{0^{'}} + \bar{\sigma}\bar{\mu}_{1^{'}} - \sqrt{2} \bar{\phi}_{0^{'}0^{'}}\lambda_0 = 0,
    \end{equation}
\end{subequations}
while the second is
\begin{subequations}
    \begin{equation}\label{eq:m-bar susy condition a}
        \ethmm\lambda_{1} + \rho{'} \lambda_{0} + (-K + \sqrt{2}\phi_{01})\bar{\mu}_{1{'}} = 0, \qquad \ethmm \bar{\mu}_{1^{'}} + \bar{\sigma}{'}\bar{\mu}_{0^{'}} - \sqrt{2}\bar{\phi}_{1^{'}1^{'}} \lambda_1 = 0,
    \end{equation}
    \begin{equation}\label{eq:m-bar susy condition b}
        \ethmm\lambda_{0} + \sigma\lambda_{1} + \sqrt{2} \phi_{00}\bar{\mu}_{0{'}} = 0, \qquad \ethmm \bar{\mu}_{0^{'}} + \rho \bar{\mu}_{1^{'}} + (K - \sqrt{2}\bar{\phi}_{0^{'}1^{'}})\lambda_0 = 0.
    \end{equation}
\end{subequations}

Supersymmetric surfaces are special and so the above equations do not admit a non-trivial solution on a general surface in a general spacetime. However, following ideas of Dougan and Mason \cite{douganquasi}, it will be convenient to impose a subset of these equations on such a surface, for which non-trivial solutions do exist. The approach of \cite{douganquasi} is to consider holomorphic Weyl spinors on $S$, defined as those in the kernel of $\bar{m}^a \nabla_a$, or anti-holomorphic spinors - those in the kernel of $m^a \nabla_a$. We will consider a gauge-covariant and supercovariant modification of a subset of the equations following from the holomorphic condition on $\bar{\lambda}_{A{'}}$ and $\bar{\mu}_{A{'}}$, given by\footnote{
An alternative approach would be to start from a holomorphic/anti-holomorphic condition on the Dirac spinor $\epsilon = (\lambda^A,\bar{\mu}_{A{'}})$, which would be equivalent to requiring $\bar{\lambda}_{A{'}}$ be anti-holomorphic/holomorphic while $\bar{\mu}_{A{'}}$ be holomorphic/anti-holomorphic. This approach has been adopted in the recent paper \cite{rallabhandi2025spinorialquasilocalmassspacetimes} which defines a quasi-local mass for spacetimes with negative cosmological constant.}
    \begin{equation} \label{eq:holomorph a}
        \bar{\ethmm}\bar{\lambda}_{1^{'}} + \rho{'} \bar{\lambda}_{0^{'}} + (-K + \sqrt{2}\bar{\phi}_{0^{'}1^{'}})\mu_1 = 0 \qquad \ethmmbar\bar{\mu}_{1^{'}} + \rho{'} \bar{\mu}_{0^{'}} + (-K - \sqrt{2}\bar{\phi}_{0^{'}1^{'}})\lambda_1 = 0,
    \end{equation}
For the main results of this paper, we will only impose (\ref{eq:holomorph a}) on surfaces with $\rho'>0$. This equation can always be solved on such a surface: it just determines $\bar{\lambda}_{0'}$ and $\bar{\mu}_{0'}$ in terms of $\bar{\lambda}_{1'}$ and $\bar{\mu}_{1'}$.

As in \cite{gibbons1982bogomolny} and \cite{reall2025lawblackholemechanics} we introduce a supercovariant generalisation of the Nester-Witten two-form \cite{nester1981new}:
\begin{equation} \label{eq:generalisation of nester-witten two-form}
    \hat{\Lambda}(\epsilon) \coloneqq \text{Re}\Big\{(-i) \Big(\lambda_B\hat{\nabla}_a\bar{\lambda}_{B^{'}} + \mu_B\hat{\nabla}_a\bar{\mu}_{B^{'}}\Big)\Big\} dx^a \wedge dx^b.
\end{equation}
Using this, we define a functional over the space of smooth spinor fields on $S$ as
\begin{equation}
\label{ISdef}
    \hat{I}_{S}[\epsilon] \coloneqq \int_S \hat{\Lambda}(\epsilon),
\end{equation}
which we henceforth refer to as the Nester-Witten functional (over $S$).
Expanding this out in our NP tetrad gives
\begin{equation} \label{eq:nw functional}
\begin{split}
    \hat{I}_{S}[\epsilon] = \text{Re} \int_S \left\{ \lambda_1 \Big[\bar{\ethmmbar} \bar{\lambda}_{0^{'}} + \rho \bar{\lambda}_{1^{'}} + (K + \sqrt{2}\bar{\phi}_{0^{'}1^{'}})\mu_0\Big] - \lambda_0 \Big[\bar{\ethmm}\bar{\lambda}_{1^{'}} + \rho{'} \bar{\lambda}_{0^{'}} + (-K + \sqrt{2}\bar{\phi}_{0^{'}1^{'}})\mu_1\Big] \right. + \\ \left.
    \mu_1 \Big[\ethmm \bar{\mu}_{0^{'}} + \rho \bar{\mu}_{1^{'}} + (K - \sqrt{2}\bar{\phi}_{0^{'}1^{'}})\lambda_0 \Big] - \mu_0 \Big[\ethmmbar\bar{\mu}_{1^{'}} + \rho{'} \bar{\mu}_{0^{'}} + (-K - \sqrt{2}\bar{\phi}_{0^{'}1^{'}})\lambda_1 \Big] \right\}.
\end{split}
\end{equation}
Following \cite{reall2025lawblackholemechanics}, we now state and prove a series of simple lemmas regarding the Nester-Witten functional.
\begin{lemma}\label{lemma:zero on susy surface}
    If $S$ is supersymmetric with spinor field $\epsilon$ then $\hat{I}_{S}[\epsilon] =0$.
\end{lemma}
\begin{proof}
    If $S$ is supersymmetric with spinor $\epsilon$, then the pull-back of $\hat{\Lambda}(\epsilon)$ to $S$ is zero by the definition of supersymmetry of $S$.
\end{proof}
\begin{lemma}\label{lemma:equiv to 2.3}
    If $\epsilon$ satisfies (\ref{eq:holomorph a}) on $S$ then 
    \begin{equation}\label{eq:nw functional when m-bar susy condition holds}
        \hat{I}_S[\epsilon] = \int_S \left[ \rho{'}\left(|\lambda_0|^2 + |\mu_0|^2 \right) + \rho \left(|\lambda_1|^2 + |\mu_1|^2 \right)\right]. 
    \end{equation}
\end{lemma}
\begin{proof}
    Given equation (\ref{eq:holomorph a}) is satisfied by $\epsilon$, both the second and fourth terms on the RHS of equation \eqref{eq:nw functional} are set to zero. 
    Then, integrating by parts using the identity (\ref{eq:integration by parts for gauge covariant eth}), one can again use equation (\ref{eq:holomorph a}) to eliminate the remaining terms involving gauge covariant derivatives. 
    All terms involving $\bar{\phi}_{0^{'}1^{'}}$ and $K$ are imaginary and so drop out when we take the real part, resulting in the above equation.
\end{proof}
\begin{lemma} \label{lemma:equiv to 2.4}
    A supersymmetric surface is not trapped.
\end{lemma}
\begin{proof}
    By definition, a trapped surface has $\rho>0$ and $\rho{'}>0$.
    Assume that $S$ is both supersymmetric and trapped.
    By supersymmetry of $S$ (of which equation (\ref{eq:holomorph a}) is a condition), Lemma \ref{lemma:zero on susy surface} tells us that the LHS of equation (\ref{eq:nw functional when m-bar susy condition holds}) is zero. 
    Since the RHS is non-negative for $S$ trapped, we see that $\epsilon$ must be identically zero on $S$ for equality.
    However, by the definition of $S$ being supersymmetric, $\epsilon$ cannot be identically zero and therefore we have a contradiction.
\end{proof}

If $S$ is marginally outer trapped and the ingoing null geodesics normal to $S$ are strictly converging then we have
\begin{lemma} \label{lemma: equiv to 2.5}
    If $S$ is supersymmetric with $\rho{'}>0$ and $\rho \equiv 0$, then $V \equiv 0$ on $S$ and $X^a$ is an outgoing null normal to $S$.
\end{lemma}
\begin{proof}
    Lemmas \ref{lemma:zero on susy surface} and \ref{lemma:equiv to 2.3} imply $\mu_0 = \lambda_ 0 \equiv 0$ on $S$, i.e., $\mu_A,\lambda_A \propto o_A$ on $S$, hence $V \equiv 0$ and $X^a \propto l^a$.
\end{proof}

\subsection{Non-negativity of the Nester-Witten functional}

Following \cite{reall2025lawblackholemechanics}, we prove the following rigidity result for a compact surface $\Sigma$ whose boundary is supersymmetric:
\begin{theorem} \label{thm:equiv to 2.6}
    Let $\Sigma$ be any smooth, compact, connected, spacelike 3-surface such that matter satisfies the local mass-charge inequality (\ref{eq:local mass-charge inequality}) on $\Sigma$.
    Assume that $\Sigma$ has boundary $\partial \Sigma = S$, where $S$ is a compact, connected 2-surface.
    Assume that the future-directed ingoing null geodesics normal to $S$ are strictly converging on $S$ - i.e. $\rho{'}>0$ on $S$. Let $\epsilon$ be a non-zero solution to (\ref{eq:holomorph a}) on $S$. 
    Then, $\hat{I}_S [\epsilon] \geq 0$ with equality if, and only if, (i) $S$ is supersymmetric with spinor $\epsilon$, (ii) $\epsilon$ extends to a spinor on $\Sigma$ satisfying $h^{b}{_a}\hat{\nabla}_b \epsilon = 0$ where $h^{b}{_a}$ is the projection onto $\Sigma$ and (iii) on $\Sigma$ the charged matter satisfies
    \begin{equation}\label{eq:equiv to eqn 19 in 1.5}
        N^a\Big( T_{ab}^{(m)} X^b - \text{Re}(V)J_a\Big) = 0,
    \end{equation}
    where $N^a$ is a normal to $\Sigma$.
\end{theorem}
\begin{proof}
    We begin by considering a spinor field, $\tilde{\epsilon}$, on $\Sigma$ satisfying the following boundary conditions:
        \begin{equation} \label{eq: boundary conditions on S in 1.5}
            \bar{\tilde{\lambda}}_{1^{'}} = \bar{\lambda}_{1^{'}} \text{ and } \bar{\tilde{\mu}}_{1^{'}} = \bar{\mu}_{1^{'}} \text{ on }\; S.
        \end{equation}
    As in \cite{douganquasi,reall2025lawblackholemechanics} (see also \cite{Rogatko_2002}), after imposing equation (\ref{eq:holomorph a}) on $\epsilon$, we can relate the integrals of the Nester-Witten two-forms of $\epsilon$ and $\tilde{\epsilon}$:
    \begin{equation} \label{eq: equiv to 21 in 1.5}
        \hat{I}_S[\epsilon] = \hat{I}_S[\tilde{\epsilon}] + \int_S \rho{'}\Big(|\tilde{\lambda}_0 - \lambda_0|^2 + |\tilde{\mu}_0 - \mu_0|^2 \Big).
    \end{equation}
    The ingoing null geodesics on $S$ are strictly converging so $\rho{'}>0$ on $S$, hence the non-negativity of $\hat{I}_{S}[\epsilon]$ follows from non-negativity of $\hat{I}_{S}[\tilde{\epsilon}]$.
    To establish the latter, we use Stokes' theorem to obtain $\hat{I}_{S}[\tilde{\epsilon}] = \int_\Sigma d\hat{\Lambda}(\tilde{\epsilon})$. In Appendix \ref{appendix:supercovariant sparling-like identity}, we show, by methods similar to that of \cite{gibbons1982bogomolny, ghhp,witten1981new}, that the RHS can be rearranged as follows:
    \begin{equation}
    \label{eq:sw_id}
    \hat{I}_S[\tilde\epsilon] = \int_\Sigma d\hat\Lambda(\tilde\epsilon) = \int_\Sigma  \Big\{ 8\pi N^a\big(T^{(m)}_{ab}X^b + J_a \text{Re}(V) \big) - 2 h^{ij}(\hat{\nabla}_i\tilde\epsilon)^{\dagger}(\hat{\nabla}_j\tilde\epsilon) - (\gamma^i\hat{\nabla}_i\tilde\epsilon)^{\dagger}(\gamma^j\hat{\nabla}_j\tilde\epsilon)\Big\}.
    \end{equation}
    Here indices $i,j$ refer to spatial basis vectors tangent to $\Sigma$.
    If the local mass-charge inequality holds on $\Sigma$, then all terms but the final term are non-negative (the induced metric $h_{ij}$ is negative-definite).
    The final term is non-positive.
    We now impose the gauge-supercovariant Sen-Witten equation on $\tilde\epsilon$ on $\Sigma$, i.e.
    \begin{equation}\label{eq:gauge-supercovariant sw}
        \gamma^i \hat{\nabla}_i \tilde\epsilon = \gamma^a h{_a}{^b}\hat{\nabla}_b \tilde\epsilon = 0,
    \end{equation}
    subject to boundary conditions (\ref{eq: boundary conditions on S in 1.5}), to obtain
    \begin{equation}\label{eq:non-negativity of nw functional on sw spinors}
        \hat{I}_S[\tilde\epsilon] = \int_\Sigma  \Big\{ 8\pi N^a\big(T^{(m)}_{ab}X^b + J_a \text{Re}(V) \big) - 2 h^{ij}(\hat{\nabla}_i\tilde\epsilon)^{\dagger}(\hat{\nabla}_j\tilde\epsilon)\Big\},
    \end{equation}
    This, therefore, proves the statement $\hat{I}_{S}[\epsilon] \geq 0$. Justification for the existence of a solution to \eqref{eq:gauge-supercovariant sw} with boundary conditions \eqref{eq: boundary conditions on S in 1.5} is given in Appendix A of \cite{reall2025lawblackholemechanics} for the case $\Lambda = 0$ and generalised in Appendix \ref{sec:existence} of this paper.
    
    If $\hat{I}_{S}[\epsilon] =0$ then from equation (\ref{eq: equiv to 21 in 1.5}) we see that as both terms are individually non-negative, they must individually be zero.
    The second term implies that $\epsilon = \tilde{\epsilon}$ on $S$ and hence $\tilde{\epsilon}$ is an extension of $\epsilon$ and we drop the tilde.
    The vanishing of the first term implies the vanishing of \eqref{eq:non-negativity of nw functional on sw spinors}, which is a sum of non-negative terms and therefore each vanishes individually.
    The second term vanishes iff $h^a{_b} \hat{\nabla}_a\epsilon = 0$ on $\Sigma$ (which implies that $S$ is supersymmetric), while the first vanishes iff equation (\ref{eq:equiv to eqn 19 in 1.5}) holds in $\Sigma$. 
    This shows that $\hat{I}_{S}[\epsilon] = 0$ implies (i), (ii) and (iii).
    Conversely, if (i), (ii) and (iii) hold then the above implies that $\hat{I}_S[\epsilon] = 0$.
\end{proof}

An immediate corollary of this result is \cite{reall2025lawblackholemechanics}
\begin{corollary}\label{corollary:equiv to 2.6.1}
    If $\Sigma$ and $S$ satisfy the assumptions of Theorem \ref{thm:equiv to 2.6} and $S$ is supersymmetric, then every smooth 2-surface contained in $\Sigma$ is also supersymmetric and hence, by Lemma \ref{lemma:equiv to 2.4}, not trapped.
\end{corollary}
\begin{proof}
    Since $S$ is supersymmetric (with spinor $\epsilon)$, this implies (by lemma \ref{lemma:zero on susy surface}) that $\hat{I}_S[\epsilon] = 0$.
    Therefore, by the above theorem, $\epsilon$ extends to a spinor such that $h^a{_b}\hat{\nabla}_a \epsilon = 0$ on $\Sigma$ and so $\epsilon$ is supercovariantly constant on any spacelike 2-surface within $\Sigma$.
    Since, $\epsilon$ is nowhere zero on $\Sigma$ (as if it were then $\epsilon$ would have to vanish everywhere on $\Sigma$) and is supercovariantly constant on any spacelike 2-surface in $\Sigma$, every spacelike 2-surface within $\Sigma$ is supersymmetric with spinor $\epsilon$.
\end{proof}


\subsection{The third law}

We can now present the main result of this section.

\begin{theorem} \label{thm:equiv to 2.8}
Consider a smooth spacetime satisfying the Einstein-Maxwell equations with a negative cosmological constant and charged matter satisfying the local mass-charge inequality (\ref{eq:local mass-charge inequality}). Let $\Sigma$ be a smooth, compact, connected, spacelike 3-surface with $\partial \Sigma=S$ where $S$ is a smooth, compact, connected,  marginally outer trapped surface such that the ``ingoing'' future-directed null geodesics normal to $S$ are strictly converging - i.e. $\rho{'}>0$ and $\rho\equiv0$ on $S$. Then $S$ cannot be supersymmetric.
\end{theorem}
This is stronger than the corresponding result for vanishing cosmological constant \cite{reall2025lawblackholemechanics}, which asserts that if $S$ is supersymmetric then the electromagnetic charges of the matter on $\Sigma$ must vanish. Theorem \ref{thm:equiv to 2.8} implies that a black hole admitting a supersymmetric horizon cross-section cannot have formed through gravitational collapse, a result we'll formulate more precisely as a corollary below,

We start with some motivation for the structure of the following proof. In the situation described, if $S$ is supersymmetric then Theorem \ref{thm:equiv to 2.6} implies that there exists a spinor satisfying $h^{b}{_a}\hat{\nabla}_b \epsilon = 0$ on $\Sigma$. This result holds for {\it any} compact 3-surface with boundary $S$. Since the domain of dependence $D(\Sigma)$ is foliated by such surfaces, this strongly suggests that $\epsilon$ extends to a supercovariantly constant spinor in $D(\Sigma)$, i.e., $\hat{\nabla}_b \epsilon = 0$. From such a spinor we can define the vector $X^a$ and scalar $V$ throughout $D(\Sigma)$, as well as a $2$-form $\Psi$, and these will satisfy certain differential relations \cite{Caldarelli_2003}, for example $X^a$ is a Killing vector field. 
The proof is then based on considering the Komar integral (over $S$) defined by this Killing field, converting to an integral over $\Sigma$ involving the matter on $\Sigma$ and then using results about the allowed form of the stress tensor of matter in supersymmetric spacetimes to reach a contradiction. In practice, the main difficulty in implementing these steps is to prove that $\epsilon$ defined on $\Sigma$ can be extended to a supercovariantly constant spinor in $D(\Sigma)$. Following \cite{reall2025lawblackholemechanics} we evade this difficulty by introducing an auxiliary spacetime, the ``Killing development'' of the hypersurface $\Sigma$ \cite{Beig_1996}. This is a spacetime defined by taking initial data on $\Sigma$ and extending off $\Sigma$ by demanding that $X^a$ extends as a Killing vector field. It then turns out that $\epsilon$ can be extended as a supercovariantly constant spinor in this spacetime \cite{Chru_ciel_2006}. The argument just described is then applied in this spacetime. 

\begin{proof}
    We assume the existence of such a $\Sigma$ with supersymmetric boundary $S$ and aim to show a contradiction. Let $(M,g)$ denote our spacetime. 
    
    We first use an argument from Appendix B of \cite{Chru_ciel_2006}, which generalises straightforwardly to include a cosmological constant. Since $S\subset M$ is supersymmetric with spinor $\epsilon$, by Theorem \ref{thm:equiv to 2.6} we know that $\epsilon$ extends to a spinor on $\Sigma$ such that $h^a{_b}\hat{\nabla}_a \epsilon = 0$. This implies that $X^a$ satisfies the tangential projection of Killing's equation onto $\Sigma$, so the components of $X^a$ on $\Sigma$ constitute ``Killing initial data''. One can then construct the ``Killing development'' $(\tilde{M}, \tilde{g})$ of the initial data on $\Sigma$. This is defined by the property that $X^a$ extends as a Killing vector field in $(\tilde{M}, \tilde{g})$. It contains a Cauchy surface diffeomorphic to $\Sigma$, which we also denote as $\Sigma$, on which the induced metric, extrinsic curvature and Maxwell field  are the same as those on $\Sigma$ in the original spacetime $(M,g)$ \cite{Beig_1996,Chru_ciel_2006}. The Maxwell field is defined throughout $(\tilde{M}, \tilde{g})$ by extending it off $\Sigma$ using Lie propagation along $X^a$. By a slight modification of the work in Appendix B of \cite{Chru_ciel_2006}, one can then show that $\epsilon$ extends to a supercovariantly constant spinor in $(\tilde{M}, \tilde{g})$.

    The existence of a supercovariantly constant spinor implies that the following equation holds in $(\tilde{M}, \tilde{g})$ \cite{Caldarelli_2003}:
    \begin{equation} \label{eq:scc!}
            \tilde{\nabla}_a X_b = \frac{K}{\sqrt{2}} \Psi_{ab} - f F_{ab} + g(\star F)_{ab}.
    \end{equation}
    where $\tilde\nabla$ is the Levi-Civita connection associated with $(\tilde{M}, \tilde{g})$, $f+ig \equiv V$ and $\Psi_{ab} \equiv i\bar{\epsilon}\gamma_{[a}\gamma_{b]}\epsilon$ is a two-form in $\tilde{M}$. We want to evaluate this equation on $S$ using the fact that Lemma \ref{lemma: equiv to 2.5} tells us that $f$ and $g$ vanish on $S$. However, as in \cite{reall2025lawblackholemechanics}, it is not clear that derivatives of $X^a$ transverse to $S$ are well-defined in $\tilde{M}$, as these directions do not lie in the Killing development of $\Sigma$.
    Instead, we consider a one-parameter family of spacelike 2-surfaces, $S_\delta\subset \text{int}(\Sigma)$ such that $\lim_{\delta\rightarrow 0} S_\delta = S$. By equation (\ref{eq:scc!}), $\lim_{\delta\rightarrow 0}dX|_{S_{\delta}}=K\sqrt{2}\Psi|_S$.
    The hypersurfaces bounded by such $S_\delta$ are $\Sigma_\delta \subset \Sigma$ with $\lim_{\delta\rightarrow 0}\Sigma_\delta= \Sigma$.
    Using the fact that $X^a$ is Killing in $\tilde{M}$, we obtain a Komar-like identity:
    \begin{equation} \label{eq:26 equivalent}
        \begin{split}
            K\sqrt{2}\int_S \star\Psi &= \lim_{\delta \rightarrow 0}\int_{S_{\delta}}\star dX = -\lim_{\delta \rightarrow 0}\int_{\Sigma_{\delta}} 2 N^a \tilde{R}_{ab}X^b = -\int_\Sigma 2N^a \tilde{R}_{ab} X^b,\\
            & = -16\pi \int_\Sigma N^a X^b \left[T_{ab}^{MW} + \left(\tilde{T}^{(m)}_{ab} - \frac{1}{2}\tilde{T}^{(m)}\tilde{g}_{ab}\right) - \frac{\Lambda}{8\pi} \tilde{g}_{ab} \right].
        \end{split}
    \end{equation}
where $N^a$ is the future-directed unit normal to $\Sigma$, $\tilde{R}_{ab}$ is the Ricci tensor of $\tilde{g}_{ab}$ and $\tilde{T}^{(m)}_{ab}$ is the matter energy-momentum tensor in $(\tilde{M},\tilde{g})$ defined by Einstein's equation with negative cosmological constant, equation \eqref{eq:efe}. The Maxwell field on $\Sigma$ is the same in both spacetimes so we don't need a tilde on $T_{ab}^{MW}$. Consider the left-hand side of the above equation in an orthonormal basis $\{e_0^a,e_i^a\}$ with $e_0^a$ and $e_3^a$ normal to $S$. This gives
    \begin{equation} \label{eq:integral of 2-form over S}
        \int_S \star\Psi = 
        \int_S \Psi_{03} d S,
    \end{equation}
    where $dS$ is the volume element on $S$.    
    The results of lemma \ref{lemma: equiv to 2.5} imply that, in Dirac spinor notation, $\gamma^0 \epsilon = \gamma^3 \epsilon$ on $S$. 
    Using this, $\Psi_{03}|_S = i\bar{\epsilon}\gamma_0\gamma_3\epsilon|_S = -i f|_S=0$. Hence (\ref{eq:integral of 2-form over S}) is zero and so \eqref{eq:26 equivalent} becomes
    \begin{equation}
    \label{eq:int_constr}
        \int_\Sigma N^a X^b \left[ T_{ab}^{MW} + \left(\tilde{T}^{(m)}_{ab} - \frac{1}{2}\tilde{T}^{(m)}\tilde{g}_{ab}\right) - \frac{\Lambda}{8\pi} \tilde{g}_{ab} \right] = 0.
    \end{equation}
    The next step is to understand the properties of $\tilde{T}^{(m)}_{ab}$.
    The form of $\tilde{T}^{(m)}_{ab}$ is highly constrained by the existence of a supercovariantly constant spinor. Generalising results of \cite{caldarelli2004supersymmetric},
    in equation \eqref{eq:result for theorem 2.8 introduction} of Appendix \ref{appendix:form of stress tensor with null killing spinor} we show the following result, true at any $p\in\tilde M$ (and so, in particular, along $\Sigma$):
    \begin{equation} \label{eq:appendix b stuff}
        (\tilde{T}^{(m)}_{ab}-\frac{1}{2}\tilde{g}_{ab}\tilde{T}^{(m)})X^b = \frac{1}{2}\tilde{T}^{(m)}_{ab}X^b.
    \end{equation}
    Using this, equation (\ref{eq:int_constr}) becomes
    \begin{equation} \label{eq:vanishing of things}
        \int_\Sigma \left( T_{ab}^{MW}N^a X^b  + \frac{1}{2}\tilde{T}^{(m)}_{ab}N^aX^b - \frac{\Lambda}{8\pi} N\cdot X\right)= 0.
    \end{equation}
    The sign of the second term in the integrand can be determed as follows. 
    On $\text{int}(\Sigma)$ the induced metric, extrinsic curvature, Maxwell field and cosmological constant are the same in both spacetimes $(M,g)$ and $(\tilde{M},\tilde{g})$. From the Hamiltonian constraint, we then have that, on $\Sigma$, $T_{ab}N^aN^b = \tilde{T}_{ab}N^aN^b$. Since the Maxwell field is the same on $\Sigma$ in both spacetimes, this reduces to $T^{(m)}_{ab}N^aN^b = \tilde{T}^{(m)}_{ab}N^a N^b$ on $\Sigma$.
    Similarly, the momentum constraint on $\Sigma$ gives $\tilde{T}^{(m)}_{ab}(e_i)^a N^b = T^{(m)}_{ab}(e_i)^a N^b$  where $\{(e_i)^a\}$ is an orthonormal basis tangential to $\Sigma$. Combining these results gives  $T^{(m)}_{ab}N^a = \tilde{T}^{(m)}_{ab}N^a$ on $\Sigma$ and hence $\tilde{T}^{(m)}_{ab}N^aX^b = T^{(m)}_{ab}N^aX^b$ on $\Sigma$.
    Since the matter stress tensor in $M$ satisfies the dominant energy condition (as it satisfies the local mass-charge inequality), the right-hand side of this is non-negative on $\Sigma$ and hence $\tilde{T}^{(m)}_{ab}N^aX^b\geq 0$.
   
    Returning to \eqref{eq:vanishing of things}, the first term is non-negative because a Maxwell field satisfies the dominant energy condition and $N^a$, $X^a$ are causal. The second term is non-negative by the above. The final term is strictly positive since $\Lambda<0$. Hence the LHS of \eqref{eq:vanishing of things} is strictly positive and we have reached a contradiction. 
\end{proof}

Note that, unlike in theorem 2.8 of \cite{reall2025lawblackholemechanics}, we cannot weaken the condition that $\rho\equiv 0$ on $S$ to $X^a$ being null on $S$. This is because the first term on the RHS of equation (\ref{eq:scc!}) does not necessarily vanish when $X^a$ is null, unlike the $K=0$ case.

We can now state the third law of black hole mechanics as the following corollary of Theorem \ref{thm:equiv to 2.8}. Consider a $2$-surface $S$ that is a cross-section of the event horizon of a black hole formed in gravitational collapse. Then we can write $S = \partial \Sigma$ for a compact spacelike surface $\Sigma$, corresponding to the black hole interior. The corollary asserts that such $S$ cannot ``look the same'' as a horizon cross-section of a supersymmetric black hole (e.g. a supersymmetric Kerr-Newman-adS black hole). Hence a supersymmetric black hole cannot form in finite time in gravitational collapse.

\begin{corollary}\label{theorem:third law in generality}
Consider a smooth spacetime $(M,g)$ satisfying the Einstein-Maxwell equations with a negative cosmological constant and charged matter satisfying the local mass-charge inequality \eqref{eq:local mass-charge inequality}. Let $\Sigma$ be a  smooth, compact, connected, spacelike 3-surface with $\partial \Sigma=S$ where $S$ is a smooth, compact, connected, $2$-surface. Let $(\tilde{M},\tilde{g})$ be a spacetime with a Maxwell field and admitting a non-trivial supercovariantly constant spinor. Assume that, within $\tilde{M}$, there exists a Killing horizon $\mathcal{H}^{+}$. Let $\tilde{S}$ be a cross-section of $\mathcal{H}^{+}$ on which the ``ingoing" future-directed null geodesics (i.e. those not tangent to $\mathcal{H}^{+}$) normal to $\tilde{S}$ are strictly converging. Then, there does not exist a diffeomorphism $\Phi$ such that:
    \begin{enumerate}
        \item $\Phi: U\rightarrow\tilde{U}$, where $U$ is a neighbourhood of $S$ in $M$ and $\tilde{U}$ is a neighbourhood of $\tilde{S}$ in $\tilde{M}$, such that $\Phi$ maps $\Sigma\cap U$ to the interior of $\mathcal{H}^{+}$ (i.e. those points in $\tilde{M}$ that can be reached by ingoing null geodesics) and that $\Phi(S)=\tilde{S}$,
        \item The pull-back map, $\Phi^\star$, maps the spacetime metric, extrinsic curvature and $U(1)$-gauge field on $\tilde{S}$ to the corresponding quantities on $S$.
    \end{enumerate}
\end{corollary}

\begin{proof}

We proceed by assuming that a diffeomorphism with the listed properties does exist, and then show that this leads to a contradiction of Theorem \ref{thm:equiv to 2.8}.

    Since $\Phi^\star$ maps the spacetime metric on $\tilde{S}$ to that on $S$, given a NP tetrad on $\tilde{S}$, we can pull this back to a NP tetrad on $S$.
    By assumption, the “ingoing" null normal to $\tilde{S}$ is mapped to an “ingoing" null normal to $S$, call it $n^a$ (and vice versa for “outgoing", which, on $S$, we refer to as $l^a$).
    Assumption 2 says that the gauge field on $\tilde{S}$ is pulled-back to the gauge field on $S$, and so, since the pull-back commutes with the exterior derivative, the Maxwell field on $\tilde{S}$ is pulled back to that on $S$.
    Similarly, the expansion and shear of “ingoing" and “outgoing" null geodesics from $S$ are precisely the same quantities pulled-back from $\tilde{S}$ (the latter quantities are all directly related to the extrinsic curvatures along the ingoing/outgoing null normals).
    Because ingoing null geodesics normal to $\tilde{S}$ are strictly converging, on $S$ we have that the integral curves of $n^a$ are strictly converging - i.e. $\rho{'}>0$ on $S$.
    Because $\tilde{S}$ is a cross-section of a Killing horizon, the expansion of “outgoing" null geodesics from $S$ is precisely zero, i.e. $\rho\equiv0$ on $S$, so $S$ is marginally outer trapped.
    
    It now remains to establish the correspondence of the $\ethmm$ operators on each surface.
    By the assumption that the spacetime metric, gauge field and ingoing null normal are pulled back from $\tilde{S}$ to $S$, all components of the connection terms
    in $\ethmm$ on $S$ are completely determined by the pull-back along $\Phi$ from $\tilde{S}$.
    To see this explicitly, the only connection terms not specified by quantities intrinsic to $S$ are the spin-coefficients $\beta$ and $\beta{'}$, which involve tangential derivatives of the basis vectors normal to $S$.
    Since, by assumption, the pull-back $\Phi^\star$ maps the NP tetrad on $\tilde{S}$ to the tetrad on $S$ and the spacetime covariant derivative on $\tilde{S}$ to that on $S$ (which follows from the assumption that the spacetime metric on $\tilde{S}$ is mapped to that on $S$), the spin-coefficients $\beta$ and $\beta^{'}$ on $S$ are the pull-back of those quantities on $\tilde{S}$.
    That the tangential components of the gauge field in $\tilde{S}$ appearing in the connection terms of $\ethmm$ are pulled back to those in $S$ is guaranteed by virtue of the assumption that the pull back of the gauge field on $\tilde{S}$ is the gauge field on $S$.
    Therefore, we see that $\ethmm$ on $S$ is wholly specified via the pull-back of quantities on $\tilde{S}$ defining the same operator on that surface.
    Note that if we had assumed that only the induced metric of $\tilde{S}$ was pulled back to the induced metric on $S$, then such an identification would not have been possible as the connection terms $\ethmm$ on $S$ would not be fully specified by quantities pulled-back along $\Phi$ from $\tilde{M}$.
    
    Given this, the conditions of supersymmetry in equations (\ref{eq:m susy condition a}), (\ref{eq:m susy condition b}), (\ref{eq:m-bar susy condition a}), (\ref{eq:m-bar susy condition b}) take the same form on $S$ as on $\tilde{S}$. Hence the existence of a non-trivial solution of these equations on $\tilde{S}$ implies the existence of a non-trivial solution on $S$, i.e., $S$ is a supersymmetric surface. We have already seen that $S$ is marginally outer trapped. We now have a contradiction with Theorem \ref{thm:equiv to 2.8}.
\end{proof}

This result makes no reference to the asymptotic structure of the spacetime. In particular it does not require any choice of boundary conditions at infinity in adS. The result excludes the formation in finite time of any supersymmetric black hole whose event horizon is a Killing horizon. In particular, it applies to asymptotically (locally) adS black hole spacetimes such as supersymmetric Kerr-Newman-adS.
\section{Generalisation with inner boundary at a trapped surface}\label{sec:trapped surfaces}

As discussed in the Introduction, we wish to generalise the third law of \cite{reall2025lawblackholemechanics}, and our $\Lambda<0$ modification of this result, to cover black holes that are not necessarily formed in gravitational collapse. The idea is to enlarge the class of spacelike surfaces $\Sigma$ considered in \cite{reall2025lawblackholemechanics} and above to allow for an inner boundary at an outer trapped surface. The presence of such a surface can be regarded as the definition of what is meant by the initial black hole being ``non-extremal'' in this formulation of the third law \cite{israel1986third}. The idea is that $\Sigma$ extends from this inner boundary at a trapped surface inside the black hole to a late-time cross-section $S$ of the event horizon. We will show that such $S$ cannot ``look the same as'' a horizon cross-section of a supersymmetric black hole. Thus the initial non-extremal black hole cannot evolve to a supersymmetric black hole in finite time. 

We will show that the results of the previous section can be extended to allow for surfaces $\Sigma$ of this type. We will concurrently carry out a similar generalization for the $\Lambda=0$ case considered in Theorem 2.6 of \cite{reall2025lawblackholemechanics}. We begin with a generalisation of Theorem \ref{thm:equiv to 2.6} to allow for such a hypersurface. For brevity, we will only include relevant changes to the proof. The main idea is based on the work of \cite{ghhp}, which generalised the spinorial proof of the positive energy theorem to allow for an inner boundary at a marginally outer trapped surface. Recall that the notions of weakly outer trapped, marginally outer trapped etc were defined in section \ref{sec:NP}.

\begin{theorem} \label{thm:equiv to 2.6 but with inner boundary}
    Let $\Sigma$ be any smooth, compact, connected, spacelike surface in a spacetime satisfying the Einstein-Maxwell equations with cosmological constant $\Lambda \le 0$ and with matter that satisfies the local mass-charge inequality\footnote{\label{fn:flat} For $\Lambda=0$ we allow magnetically charged matter, in which case the local mass-charge inequality involves the magnetic current \cite{reall2025lawblackholemechanics}.} (\ref{eq:local mass-charge inequality}) on $\Sigma$.
    Assume that $\Sigma$ has boundary $\partial \Sigma = S\cup T$ where each connected component of $T$ is a weakly outer trapped surface. Assume that the ``ingoing'' (directed into $\Sigma$) future-directed null geodesics orthogonal to $S$ are strictly converging and let $\epsilon$ be a non-zero solution to (\ref{eq:holomorph a}) on $S$. Then, $\hat{I}_S [\epsilon] \geq 0$ with equality if, and only if,  (i) $S$ is supersymmetric with spinor $\epsilon$, (ii) $\epsilon$ extends to a spinor on $\Sigma$ satisfying $h^{b}{_a}\hat{\nabla}_b \epsilon = 0$ where $h_{ab}$ is the induced metric on $\Sigma$ (iii) the charged matter satisfies \eqref{eq:equiv to eqn 19 in 1.5} on $\Sigma$ where $N^a$ is a normal to $\Sigma$ and (iv) each connected component of $T$ is a marginally outer trapped surface.
\end{theorem}
\begin{proof}
We'll give the proof for $\Lambda<0$. The proof for $\Lambda=0$ is obtained by setting $K=0$ in the the supercovariant derivative, which alters the form of equation (\ref{eq:holomorph a}) but this has little impact on the results.

We choose our NP tetrad so that $l^a$ is tangent to the outgoing future-directed null geodesics normal to $T$ and $S$. The weakly outer trapped condition is then $\rho \geq 0$ on $T$. The condition that the ingoing null geodesics from $S$ are strictly converging is $\rho'>0$ on $S$. 

We follow the steps of the proof of Theorem \ref{thm:equiv to 2.6}, introducing another spinor field $\tilde{\epsilon}$ on $\Sigma$ satisfying the boundary conditions \eqref{eq: boundary conditions on S in 1.5} on $S$ and using equation \eqref{eq:holomorph a} to deduce equation \eqref{eq: equiv to 21 in 1.5}, which we repeat here:
\begin{equation} \label{eq: equiv to 21}
        \hat{I}_S[\epsilon] = \hat{I}_S[\tilde{\epsilon}] + \int_S \rho{'}\Big(|\tilde{\lambda}_0 - \lambda_0|^2 + |\tilde{\mu}_0 - \mu_0|^2 \Big).
    \end{equation}
Since $\rho{'}>0$ on $S$, the non-negativity of $\hat{I}_{S}[\epsilon]$ follows if we can establish non-negativity of $\hat{I}_{S}[\tilde{\epsilon}]$. To do so, we use Stokes' theorem to obtain
    \begin{equation}\label{eq:stokes with apparent horizon}
        \hat{I}_{S}[\tilde{\epsilon}] = \int_\Sigma d\hat{\Lambda}(\tilde{\epsilon}) - \hat{I}_T[\tilde\epsilon],
    \end{equation}
where $\hat{I}_T$ is defined by replacing $S$ with $T$ in \eqref{ISdef}. An explicit form for $\hat{I}_T[\epsilon]$ can be deduced from the expression \eqref{eq:nw functional} for $\hat{I}_S[\epsilon]$ as follows. The null geodesics directed {\it into} $\Sigma$ have tangent $l^a$ on $T$ and $n^a$ on $S$. Hence the form of the boundary term arising on $T$ can be deduced from the boundary term on $S$ by swapping $l^a$ and $n^a$ and (to give a correctly oriented basis) also swapping $m^a$ and $\bar{m}^a$. This is the same as the GHP ``priming'' map which transforms the spin-basis vectors as $o^A\rightarrow i\iota^A$ and $\iota^A\rightarrow io^A$. Applying this operation to \eqref{eq:nw functional} gives 
\begin{equation}
    \begin{split}
    -\hat{I}_{T}[\tilde\epsilon] = \text{Re} \int_T \left\{ \lambda_1 \Big[\bar{\ethmmbar} \bar{\lambda}_{0^{'}} + \rho \bar{\lambda}_{1^{'}} + (K + \sqrt{2}\bar{\phi}_{0^{'}1^{'}})\mu_0\Big] - \lambda_0 \Big[\bar{\ethmm}\bar{\lambda}_{1^{'}} + \rho{'} \bar{\lambda}_{0^{'}} + (-K + \sqrt{2}\bar{\phi}_{0^{'}1^{'}})\mu_1\Big] \right. + \\ \left.
    \mu_1 \Big[\ethmm \bar{\mu}_{0^{'}} + \rho \bar{\mu}_{1^{'}} + (K - \sqrt{2}\bar{\phi}_{0^{'}1^{'}})\lambda_0 \Big] - \mu_0 \Big[\ethmmbar\bar{\mu}_{1^{'}} + \rho{'} \bar{\mu}_{0^{'}} + (-K - \sqrt{2}\bar{\phi}_{0^{'}1^{'}})\lambda_1 \Big] \right\},
\end{split}
\end{equation}
where we have dropped tildes on Weyl spinor components for readability.
Next, motivated by \cite{ghhp}, we assume that the spinor $\tilde{\epsilon}$ satisfies the following boundary condition on $T$:
 \begin{equation} \label{eq:boundary conditions on apparent horizon}
            \bar{\tilde{\lambda}}_{0^{'}} = \bar{\tilde{\mu}}_{0^{'}} = 0 \text{ on } T
\end{equation}
which gives
 \begin{equation}\label{eq:nw functional on H}
        -\hat{I}_{T}[\tilde{\epsilon}] =  \int_T\rho (|\tilde{\lambda}_1|^2 + |\tilde{\mu}_1|^2)
\end{equation}
which is manifestly non-negative since $\rho \ge 0$ on $T$.

Similarly to Theorem \ref{thm:equiv to 2.6}, the first term on the RHS of equation (\ref{eq:stokes with apparent horizon}) can be shown to be non-negative if $\tilde{\epsilon}$ satisfies the supercovariant modification of the Sen-Witten equation\footnote{\label{fn:non-negativity of nw functional on sw spinors lambda=0} The results of Appendix \ref{appendix:supercovariant sparling-like identity} don't directly apply to the $\Lambda=0$ case as in that Appendix we set the magnetic current to zero - required for consistency in the $\Lambda < 0$ case. However, using the results of \cite{gibbons1982bogomolny}, the above argument runs through identically allowing for magnetic currents.} $\gamma^b h^a{_b}\hat{\nabla}_a \tilde{\epsilon} = 0$ on $\Sigma$
with boundary conditions \eqref{eq: boundary conditions on S in 1.5} and (\ref{eq:boundary conditions on apparent horizon}). 
Justification for the existence of such a solution is given in Appendix \ref{sec:existence}.
This, therefore, proves the statement $\hat{I}_{S}[\epsilon] \geq 0$.

Now assume that $\hat{I}_{S}[\epsilon] = 0$. From the results above, both terms on the RHS of (\ref{eq:stokes with apparent horizon}) are non-negative and so they must both individually be zero. The first term in equation (\ref{eq:stokes with apparent horizon}) can be shown, in a manner identical to that of Theorem \ref{thm:equiv to 2.6}, to imply conditions (i), (ii) and (iii) hold. If $\rho > 0$ at some point of $T$ then the vanishing of the second term on the RHS of (\ref{eq:stokes with apparent horizon}) implies that $\tilde{\lambda}_1 = \tilde{\mu}_1 = 0$ at this point, and hence $\tilde{\epsilon}=0$ at this point. But, by (ii), $\tilde{\epsilon}$ is supercovariantly constant on $\Sigma$, so this implies (see e.g. \cite{reall2025lawblackholemechanics}) that $\tilde{\epsilon}$ vanishes everywhere on $\Sigma$. However $\tilde{\epsilon}=\epsilon$ on $S$ (by (ii)), which contradicts the fact that $\epsilon$ is non-zero on $S$. So we must actually have $\rho \equiv 0$ on $T$, which proves (iv). Conversely, if we assume (i) to (iv) then reversing these steps gives $\hat{I}_S[\epsilon]=0$.
\end{proof}

An example of a surface satisfying (i) to (iv) is a surface $\Sigma$ in a supersymmetric adS black hole spacetime where the inner boundary $T$ is a cross-section of the event horizon and the outer boundary $S$ lies strictly outside the black hole and is convex in the sense that ingoing null geodesics normal to $S$ are converging. 

The generalisation of corollary \ref{corollary:equiv to 2.6.1} is 
\begin{corollary}
\label{cor:equiv to 2.6.1 but with inner boundary}
    If $\Sigma$, $S$ and $T$ satisfy the assumptions of Theorem \ref{thm:equiv to 2.6 but with inner boundary} and $S$ is supersymmetric, then $T$ is marginally outer trapped and supersymmetric, and every smooth 2-surface contained in $\Sigma$ is also supersymmetric. 
\end{corollary}
\begin{proof}
 Since $S$ is supersymmetric, $\hat{I}_S[\epsilon]=0$ and $T$ is marginally outer trapped by point (iv) of Theorem \ref{thm:equiv to 2.6 but with inner boundary}. The proof of the final statement is the same as in corollary \ref{corollary:equiv to 2.6.1}. Applying this statement to $T$ shows that $T$ is supersymmetric. 
\end{proof}

We can now give the generalisation of our main theorem, and the corresponding result of \cite{reall2025lawblackholemechanics}, to include an inner boundary at a weakly outer trapped surface. 

\begin{theorem} \label{thm:equiv to 2.8 but with inner boundary}
Consider a smooth spacetime satisfying the Einstein-Maxwell equations with a cosmological constant $\Lambda \le 0$ and charged matter satisfying the local mass-charge inequality (\ref{eq:local mass-charge inequality}).
Let $\Sigma$ be a smooth, compact, connected, spacelike surface with $\partial \Sigma=S \cup T$ where $S$, the ``outer'' boundary of $\Sigma$, is a smooth, compact, connected, marginally outer trapped surface and where $T$, the ``inner'' boundary of $\Sigma$, is a finite union of weakly outer trapped 2-surfaces. Assume that the ``ingoing'' future-directed null geodesics normal to $S$ are strictly converging. (i) If $\Lambda < 0$ then $S$ cannot be supersymmetric; (ii) if $\Lambda =0$ and $S$ is supersymmetric then the electric and magnetic charge densities of matter on $\Sigma$ must vanish and $T$ must be marginally outer trapped and supersymmetric.
\end{theorem}

\begin{proof}
   The first part of the proof covers both $\Lambda<0$ and $\Lambda=0$. We will split these two cases at the end. 

   We start by assuming that $S$ is supersymmetric with spinor $\epsilon$. By Theorem \ref{thm:equiv to 2.6 but with inner boundary} (or the analogous result in \cite{reall2025lawblackholemechanics}) we know that $\epsilon$ extends to a spinor on $\Sigma$ such that $h^a{_b}\hat{\nabla}_a \epsilon = 0$ and, by Corollary \ref{cor:equiv to 2.6.1 but with inner boundary} that $T$ must be marginally  outer trapped and supersymmetric. Equation \eqref{eq:boundary conditions on apparent horizon} in the proof of Theorem \ref{thm:equiv to 2.6 but with inner boundary} shows that $\epsilon$ satisfies $\lambda_0=\mu_0=0$ on $T$ (dropping the tilde since that proof shows that $\tilde{\epsilon}$ is an extension of $\epsilon$). This implies that $V \equiv 0$ on $T$. Lemma \ref{lemma: equiv to 2.5} tells us that $V \equiv 0$ on $S$.
   
    As in the proof of Theorem \ref{thm:equiv to 2.8}, we consider the Killing development of $\Sigma$ along $X^a$ and obtain a spacetime $(\tilde{M},\tilde{g})$ wherein $X^a$ is Killing. Since $V=f+ig$ vanishes on $S$ and on $T$, we can proceed as in the proof of Theorem \ref{thm:equiv to 2.8} to obtain
    \be \label{eq:int_constr_T}
K\sqrt{2}\int_{S \cup T} \star\Psi  = -16\pi \int_\Sigma N^a X^b \left[T_{ab}^{MW} + \left(\tilde{T}^{(m)}_{ab} - \frac{1}{2}\tilde{T}^{(m)}\tilde{g}_{ab}\right) - \frac{\Lambda}{8\pi} \tilde{g}_{ab} \right].
    \ee
For $\Lambda=0$ the LHS is zero (as $K=0$). For $\Lambda<0$, the integral of $\star \Psi$ over $S$ vanishes as explained in the proof of Theorem \ref{thm:equiv to 2.8}. Its integral over $T$ vanishes by the same argument (the condition $\lambda_0=\mu_0=0$ is the same as the condition that the spinor satisfies on $S$ following from Lemma \ref{lemma: equiv to 2.5}). We have now shown that the LHS above is zero.

    In the $\Lambda<0$ case, the remainder of the proof now follows identically to that of Theorem \ref{thm:equiv to 2.8}, resulting in the following equation
    \begin{equation}
        \int_\Sigma \left( T_{ab}^{MW}N^a X^b  + \frac{1}{2}\tilde{T}^{(m)}_{ab}N^aX^b - \frac{\Lambda}{8\pi} N\cdot X\right)  = 0,
    \end{equation}
    where we have invoked the result of equation \eqref{eq:result for theorem 2.8 introduction} of Appendix \ref{appendix:form of stress tensor with null killing spinor}.
    As before, using that the constraint equations on $\Sigma$ are the same in both spacetimes, one can show that $\tilde{T}^{(m)}_{ab}N^aX^b$ is non-negative. 
    Since $\Lambda<0$, the final term on the LHS of the above expression is strictly positive and the other terms are non-negative. Thus, we have a contradiction and, so, (i) is proved.

    If $\Lambda=0$ then similarly to \cite{reall2025lawblackholemechanics} (using results of \cite{tod1983all}), at any point $p\in\tilde{M}$ we have\footnote{$\tilde{J}^a$ denotes the magnetic current, which we include for $\Lambda=0$ but not for $\Lambda<0$.}
    \begin{equation}
       \tilde{T}_{ab} = \chi X_a X_b, \qquad \tilde{J}^a = \chi\text{Re}(V)X^a, \qquad\tilde{\tilde{J}}^a = - \chi\text{Im}(V)X^a.
    \end{equation}
    where the constraint equations  imply $\chi \ge 0$ on $\Sigma$ \cite{reall2025lawblackholemechanics}. Using this in \eqref{eq:int_constr_T}, we get
    \begin{equation}
        \int_\Sigma \left( T_{ab}^{MW}N^a X^b  + \frac{1}{2} \chi X^2 N\cdot X\right) = 0,
    \end{equation}
    But this equation implies $\chi X^2 = \chi V  = 0$ and so both the electric and magnetic current densities on $\Sigma$ vanish in $(\tilde{M},\tilde{g})$. Since the Maxwell fields and the initial data on $\Sigma$ are the same in both spacetimes, the Maxwell constraint equations then imply that the electric and magnetic charge densities must vanish on $\Sigma$ in the physical spacetime. This concludes the proof of (ii).
\end{proof}

The third law is now a corollary to the above theorem, generalising Corollary \ref{theorem:third law in generality} to allow $\Sigma$ to possess an inner boundary at an outer trapped surface. This corollary considers a situation in which $\Sigma$ is a spacelike 3-surface with an inner boundary at an outer trapped surface and an outer boundary at $S$. It shows that if the local mass-charge inequality is satisfied then $S$ cannot ``look the same as'' a horizon cross-section of a supersymmetric black hole. This captures the idea that a non-extremal black hole cannot become supersymmetric in finite time.

\begin{corollary}
    Consider a smooth spacetime $(M,g)$ satisfying the Einstein-Maxwell equations with a cosmological constant $\Lambda \le 0$ and charged matter satisfying the local mass-charge inequality. Let $\Sigma$ be a  smooth, compact, connected, spacelike surface with $\partial \Sigma=S \cup T$ where $S$, the ``outer'' boundary of $\Sigma$, is a smooth, compact, connected, surface and where $T$, the ``inner'' boundary of $\Sigma$, is a finite union of outer trapped 2-surfaces. Let $(\tilde{M},\tilde{g})$ be a spacetime with a Maxwell field and admitting a non-trivial supercovariantly constant spinor. Assume that, within $\tilde{M}$, there exists a Killing horizon $\mathcal{H}^{+}$. Let $\tilde{S}$ be a cross-section of $\mathcal{H}^{+}$ on which the ``ingoing" future-directed null geodesics (i.e. those not tangent to $\mathcal{H}^{+}$) normal to $\tilde{S}$ are strictly converging.
Then, there does not exist a diffeomorphism $\Phi$ such that:
    \begin{enumerate}
        \item $\Phi: U\rightarrow\tilde{U}$, where $U$ is a neighbourhood of $S$ in $M$ and $\tilde{U}$ is a neighbourhood of $\tilde{S}$ in $\tilde{M}$, such that $\Phi$ maps $\Sigma\cap U$ to the interior of $\mathcal{H}^{+}$ (i.e. those points in $\tilde{M}$ that can be reached by ingoing null geodesics) and that $\Phi(S)=\tilde{S}$,
        \item The pull-back map, $\Phi^\star$, maps the spacetime metric, extrinsic curvature and $U(1)$-gauge field on $\tilde{S}$ to the corresponding quantities on $S$.
    \end{enumerate}
\label{cor:3rd law with inner boundary}
    \end{corollary}
\begin{proof}
Assume the existence of such a diffeomorphism and follow the proof of Corollary \ref{theorem:third law in generality}. As in that proof, the diffeomorphism is such that all important quantities have a consistent pull-back to $S$ from $\tilde{S}$ which imply that $S$ must be supersymmetric and marginally outer trapped. We then immediately obtain a contradiction to Theorem \ref{thm:equiv to 2.8 but with inner boundary}.
\end{proof}
The assumption on $T$ can be weakened slightly. If $\Lambda<0$ the result holds assuming only that the connected components of $T$ are {\it weakly} outer trapped. If $\Lambda=0$ it holds assuming that the components of $T$ are weakly outer trapped and at least one of them is not marginally outer trapped.

\section{Discussion}

We have extended in two ways the third law of black hole mechanics for supersymmetric black holes proved in \cite{reall2025lawblackholemechanics}. First, we have extended the result to supersymmetric black holes with negative cosmological constant, proving that such black holes cannot form in finite time in gravitational collapse of matter satisfying the local mass-charge inequality \eqref{eq:localbps}. Second, we have generalised the results for zero or negative cosmological constant to cover black holes that are not formed in gravitational collapse, such as two-sided black holes. Here we proved that if the initial black hole contains a trapped surface then it cannot evolve to a supersymmetric black hole in finite time if matter satisfies \eqref{eq:localbps}.

Our results apply to a black hole solution admitting a horizon cross-section $S$ on which the (spacetime) metric, Maxwell field and extrinsic curvature coincide with those of a supersymmetric black hole. By continuity, this will be the case if the metric and Maxwell field just outside $S$ agree with those of a supersymmetric black hole. But there might be more general solutions, differing from a supersymmetric black hole just outside $S$, with the metric and Maxwell field on $S$ agreeing with those of a supersymmetric black hole, but not the extrinsic curvature, in particular not the quantity $\rho'$ measuring the convergence of ingoing null geodesics from $S$. In this case, $S$ need not be supersymmetric and our results do not exclude the existence of such solutions. 

Not all extremal black holes are supersymmetric so our results do not exclude the possibility of violating the third law by forming an extremal but non-supersymmetric black hole in finite time. Since extremal Reissner-Nordstr\"om-adS is not supersymmetric, it would be interesting to know if the methods of \cite{kehle2025gravitational,kehle2024extremalblackholeformation} could be adapted to construct third-law violating solutions describing the formation in finite time of an extremal Reissner-Nordstr\"om-adS black hole. This seems very likely for matter that does not satisfy \eqref{eq:localbps}, such as a massless charged scalar \cite{kehle2025gravitational} (or small mass charged Vlasov matter \cite{kehle2024extremalblackholeformation}), but it might also be possible for matter that does satisfy \eqref{eq:localbps}.

Similarly, if matter does not satisfy \eqref{eq:localbps} then it might be possible to form a supersymmetric Kerr-Newman-adS black hole in finite time. However, constructing such a solution is probably more difficult than constructing a type of solution whose existence was conjectured in \cite{kehle2025gravitational}, describing the formation in vacuum of an extremal Kerr black hole in finite time. Progress towards constructing such solutions has been made in \cite{Kehle:2023eni}, where it is proved that a slowly rotating Kerr black hole can form in finite time. 

If matter does satisfy \eqref{eq:localbps} then our results do not exclude the possibility of forming a supersymmetric black hole in {\it infinite} time, e.g., a solution describing gravitational collapse to form a black hole for which the spacetime geometry and Maxwell field near the event horizon asymptote to those of extremal Reissner-Norstr\"om (with $\Lambda=0$) or supersymmetric Kerr-Newman-adS at infinite time. It seems plausible that such solutions do exist, and that some such solutions will be {\it critical} solutions, i.e., solutions that, in the moduli space of solutions, lie on the boundary separating solutions describing gravitational collapse to form black holes from solutions that disperse \cite{kehle2024extremalblackholeformation}. As noted in \cite{kehle2024extremalblackholeformation}, critical solutions cannot contain trapped surfaces and so they are not third-law violating solutions. Solutions approaching extremal Reissner-Norstr\"om (with $\Lambda=0$) asymptotically have been constructed for massless uncharged scalar field matter \cite{Murata:2013daa,Angelopoulos:2024yev}: these describe perturbations of a pre-existing black hole, rather than formation of a black hole by gravitational collapse. Nevertheless, they are still critical solutions, separating black hole solutions from (geodesically incomplete) non-black-hole solutions \cite{Angelopoulos:2024yev}.

\subsection*{Acknowledgments}
AMM is supported by an STFC studentship and HSR is supported by STFC grant no. ST/X000664/1.

\appendix
\section{Non-negativity of the Nester-Witten functional} \label{appendix:supercovariant sparling-like identity}

In this Appendix we will give the derivation of \eqref{eq:sw_id}. 
In this section, we choose to work with the Hodge dual of what we have called the Nester-Witten two-form.
The motivation for this is such that a clear comparison can be made with the calculation presented in \cite{Kosteleck__1996}.
Explicitly, we define the 2-form
\begin{equation}
    -\hat{E}(\psi) \coloneqq \star \hat{\Lambda}(\psi),
\end{equation}
which, written in terms of Dirac spinors (after raising indices), is
\begin{equation}
    \hat{E}^{ab}(\psi) = \bar{\psi}\gamma^{abc}\hat{\nabla}_c\psi + c.c.,
\end{equation}
where $\gamma^{abc}\coloneqq \gamma^{[a}\gamma^{b}\gamma^{c]}$.
By using properties of the Hodge dual, we have that
\begin{equation}
    \star \hat{E}(\psi) = \hat{\Lambda}(\psi),
\end{equation}
and so
\begin{equation}
    \int_{\partial\Sigma} \hat{\Lambda} = \int_\Sigma d\star{\hat{E}}.
\end{equation}
This can be written as
\begin{equation} \label{eq:nw functional in terms of divergence}
    \int_{\partial\Sigma} \hat\Lambda = \int_\Sigma N_a\nabla_b \hat{E}^{ab},
\end{equation}
where $N^a$ is the unit normal to $\Sigma$ in $M$.

Explicitly expanding the terms in the gauge-supercovariant derivatives, we write the quantity $\hat{E}$ as
\begin{equation}
    \hat{E}^{ab} = E^{ab} + T^{ab} + M^{ab} + H^{ab},
\end{equation}
where
\begin{equation}
    \begin{split}
        E^{ab} &\coloneqq \bar{\psi}\gamma^{abc}\nabla_c\psi + c.c.,\\
        P^{ab} &\coloneqq -\frac{2i}{l} A_c \bar{\psi}\gamma^{abc}\psi,\\
        M^{ab} &\coloneqq \frac{2i}{l} \bar\psi\gamma^{ab}\psi,\\
        H^{ab} &\coloneqq \frac{1}{2} F_{de}\bar{\psi}\gamma^{abc}\gamma^{de}\gamma_c\psi,
    \end{split}
\end{equation}
where we have introduced $1/l \coloneqq K/\sqrt{2}$.
As in \eqref{eq:nw functional in terms of divergence}, we now calculate the divergences of each term: we shall only calculate $E^{ab}$ and $P^{ab}$ explicitly.
Using a result of \cite{nester1981new} gives
\begin{equation}
    \nabla_a E^{ab} = 2 \nabla_a \bar{\psi} \gamma^{abc} \nabla_c \psi - G^{bc}X_c = 2 \nabla_a \bar{\psi} \gamma^{abc} \nabla_c \psi  - \Lambda X^b - 8\pi T^{(m)}{^{ba}}X_a  - 8 \pi T^{MW}{^{ba}}X_a,
\end{equation}
where $X^a$ is equivalent to \eqref{eq:vector formed from spinors} and where we have used the Einstein equation \eqref{eq:efe}.

We now calculate the divergence of $P^{ab}$.
The reason for choosing to compute this term in particular is that it was likely omitted in the results of \cite{Kosteleck__1996} due to their lack of a gauge-covariant supercovariant derivative: we will see that, without this term, non-negativity of the Nester-Witten functional cannot be proven.
\begin{equation}
    \nabla_a P^{ab} = \frac{-2i}{l}\Big\{ -\frac{1}{2}F_{ac} \bar\psi \gamma^{acb}\psi + A_c \nabla_a\bar\psi\gamma^{abc}\psi + A_c \bar\psi\gamma^{abc}\nabla_a\psi  \Big\},
\end{equation}
where we have used the definition $F=dA$.
By computing the divergences of all other terms and summing, one sees that
\begin{equation}
    \nabla_a \hat{E}^{ab} - 2 \overline{\hat{\nabla}_a\psi}\gamma^{abc}\hat\nabla_c\psi = -8\pi T^{(m)}{^b}{_a}X^a - 8\pi \bar\psi J^b \psi,
\end{equation}
where $J_b$ is the electric current of the charged matter.
In getting to this result, we have used Maxwell's equations and various identities of gamma matrices presented in \cite{gibbons1982bogomolny}.
Of key importance is that, with the inclusion of the correct gauging term in the supercovariant derivative, the correctly modified Nester-Witten two-form gains the $P_{ab}$ term: only with this term do the various terms involving products of the Maxwell field and the cosmological constant cancel.
Thus, using this, we see
\begin{equation}
    \int_{\partial\Sigma} \hat\Lambda = \int_\Sigma N^a \Big\{ 8\pi \big(T^{(m)}_{ab}X^b + J_a \text{Re}(V) \big) - 2\overline{\hat{\nabla}_a\psi}\gamma^{abc}\hat\nabla_c\psi \Big\},
\end{equation}
where we have used that $\bar\psi\psi = \text{Re}(V)$.
Finally, expanding the term quadratic in gauge-supercovariant derivatives,
\begin{equation}
    \int_{\partial\Sigma} \hat\Lambda = \int_\Sigma  \Big\{ 8\pi N^a\big(T^{(m)}_{ab}X^b + J_a \text{Re}(V) \big) - 2 h^{ij}(\hat{\nabla}_i\psi)^{\dagger}(\hat{\nabla}_j\psi) - (\gamma^i\hat{\nabla}_i\psi)^{\dagger}(\gamma^j\hat{\nabla}_j\psi)\Big\}.
\end{equation}
We see that, if the local mass-charge inequality holds on $\Sigma$, all terms but the final term are non-negative (the induced metric $h_{ij}$ is negative-definite).
The final term is non-positive.
If, now, we impose the gauge-supercovariant Sen-Witten equation on $\Sigma$, i.e.
\begin{equation}
    \gamma^i \hat{\nabla}_i \psi = \gamma^a h{_a}{^b}\hat{\nabla}_b \psi = 0,
\end{equation}
then the final term vanishes and hence
\begin{equation} \label{eq:nw 2-form integral over bulk}
    \int_{\partial\Sigma} \hat\Lambda = \int_\Sigma  \Big\{ 8\pi N^a\big(T^{(m)}_{ab}X^b + J_a \text{Re}(V) \big) - 2 h^{ij}(\hat{\nabla}_i\psi)^{\dagger}(\hat{\nabla}_j\psi)\Big\},
\end{equation}
is non-negative, precisely the result required for theorem \ref{thm:equiv to 2.6}. 
The vanishing of the two terms on the right-hand side of \eqref{eq:nw 2-form integral over bulk} is necessary and sufficient for the vanishing of the LHS.
The vanishing of the first is equivalent to \eqref{eq:equiv to eqn 19 in 1.5}, while the vanishing of the second implies that $\psi$ is supercovariantly constant over $\Sigma$.


\section{Existence of a solution to the Sen-Witten equation} \label{sec:existence}

Following \cite{reall2025lawblackholemechanics}, we present a justification for the existence of a solution to the supercovariant Sen-Witten equation on a surface $\Sigma$ with $\partial \Sigma = S\cup T$ satisfying boundary conditions on $S$ and $T$ given by equations (\ref{eq: boundary conditions on S in 1.5}) and (\ref{eq:boundary conditions on apparent horizon}) respectively. 

First we establish uniqueness.
Given two solutions to the modified Sen-Witten equation, $\epsilon_1$ and $\epsilon_2$, satisfying the boundary conditions on $S$ and $T$ given in equations \eqref{eq: boundary conditions on S in 1.5} and \eqref{eq:boundary conditions on apparent horizon} respectively, consider the difference of these two solutions, $\epsilon_{\Delta}$, which is itself a solution (satisfying homogenous boundary conditions).
Consider, now, the Nester-Witten functional over $\partial \Sigma$ acting  on $\epsilon_\Delta$: $\hat{I}_{\partial \Sigma}[\epsilon_\Delta] = \hat I_{S}[\epsilon_\Delta] + \hat I_T [\epsilon_\Delta] $. 
The results of the above Appendix show that $\hat I_{\partial \Sigma}[\tilde{\epsilon}]$ is non-negative for any $\tilde{\epsilon}$ satisfying the gauge-supercovariant Sen-Witten equation on $\Sigma$, i.e. $\hat I_{\partial \Sigma}[\epsilon_\Delta]\geq0$.
However, by inspection of equations \eqref{eq:nw functional} and \eqref{eq:nw functional on H}, we see that each of $\hat I_{S}[\epsilon_\Delta]$ and $\hat I_T [\epsilon_\Delta]$ are non-positive by virtue of the homogeneous boundary conditions of $\epsilon_\Delta$ (on $T$, the result follows as $\rho\geq 0$).
Therefore, each term is individually zero: in particular $\hat I_S [\epsilon_\Delta] = 0$ and so we have that $\epsilon_\Delta = 0$, as can be seen by inspection of \eqref{eq:nw functional}.
The results of the above appendix imply that for $\hat I_{\partial \Sigma} [\epsilon_\Delta] = 0$ then $h^a{_b}\hat{\nabla}_a\epsilon_{\Delta} = 0$ on $\Sigma$ and so $\epsilon_\Delta$ must be zero everywhere (as it is zero on $S$).
In particular, $\epsilon_\Delta = 0$ on $T$, and so $\hat I_T [\epsilon_\Delta] = 0$ for $T$ weakly outer trapped as required.
Therefore, solutions to the gauge-supercovariant Sen-Witten equation, with boundary conditions (\ref{eq: boundary conditions on S in 1.5}) and (\ref{eq:boundary conditions on apparent horizon}), are unique.
This establishes that the gauge-supercovariant Sen-Witten operator with homogeneous boundary conditions has trivial kernel.

We now establish existence. Following \cite{reall2025lawblackholemechanics}, we now provide a heuristic argument for why, as stated in a footnote of \cite{douganquasi}, checking that the kernel of the adjoint problem is trivial is sufficient to show existence.
Firstly, we convert the boundary conditions satisfied by spinor $\epsilon$ on $S$ to homogenous ones by subtracting off an arbitrary, smooth spinor field, $\epsilon_S$, satisfying the conditions (\ref{eq: boundary conditions on S in 1.5}) on $S$ and (\ref{eq:boundary conditions on apparent horizon}) on $T$.
We set up an orthonormal basis on $\Sigma$, such that the $0$-direction is normal to $\Sigma$ and the $3$-direction is normal to spacelike 2-surfaces lying within $\Sigma$.
The homogenous boundary conditions on $S$ and $T$, can be written in Dirac spinor form as:
\begin{equation} \label{eq:boundary conditions in dirac spinor notation}
    (\gamma^0 + \gamma^3)\epsilon = 0 \text{ on } S,\qquad (\gamma^0-\gamma^3)\epsilon = 0 \text{ on }T.
\end{equation}
The gauge-supercovariant Sen-Witten equation is now inhomogenous with homogenous boundary conditions given by equation (\ref{eq:boundary conditions in dirac spinor notation}) above:
\begin{equation} \label{eq:sw existence stuff}
    \hat{\mathcal{D}}\epsilon \coloneqq \gamma^b h^a{_b}\hat{\nabla}_a\epsilon = f,
\end{equation}
where $f\coloneqq-\hat{\mathcal{D}}\epsilon_S$.
On the space of charged Dirac spinors on $\Sigma$, we introduce the following inner product:
\begin{equation}
    (\eta,\tau)\coloneqq \int_\Sigma \eta^\dagger \tau.
\end{equation}
We solve equation (\ref{eq:sw existence stuff}) by requiring that $\epsilon$ be such that the norm (defined by the above inner product) $||\hat{\mathcal{D}}\epsilon -f||^2$ is minimised.
Let $\eta\coloneqq \hat{\mathcal{D}}\epsilon-f$.
Varying with respect to $\epsilon$, we get
\begin{equation}\label{eq:variations}
    \delta_\epsilon(\eta,\eta) = \int_\Sigma\eta^\dagger \hat{\mathcal{D}}\delta\epsilon = \int_\Sigma \prescript{3}{}{\hat{\nabla}}_i(\eta^\dagger \gamma^i \delta\epsilon) + (\hat{\mathcal{D}}^\dagger\eta)^\dagger\delta\epsilon,
\end{equation}
where $\prescript{3}{}{\hat{\nabla}}_a$ is the the induced covariant derivative in $\Sigma$.
The total derivative term can be converted to integrals over $S$ and $T$ by use of Stokes' theorem.
The integrands of the boundary terms are both of the form $\eta^\dagger n_a\gamma^a\delta\epsilon$, where $n^a$ is the unit normal to each surface respectively and, in our orthonormal basis, has only a $3$-component ($n^a$ from $S$ points into $\Sigma$, while $n^a$ from $T$ points out of $\Sigma$).
These surface terms need to vanish for a well-defined variational problem.
Whence, for variations $\delta\epsilon$ satisfying (\ref{eq:boundary conditions in dirac spinor notation}) we have
\begin{equation}
    \eta^{\dagger}\gamma^3\delta\epsilon = \frac{1}{2} \eta^\dagger[(\gamma^0+\gamma^3)-(\gamma^0-\gamma^3)]\delta\epsilon=0.
\end{equation}
Upon using the boundary conditions satisfied by $\delta\epsilon$ on $S$, this implies
\begin{equation}
    \eta^\dagger(\gamma^0-\gamma^3)=0 \implies (\gamma^0+\gamma^3)\eta = 0,
\end{equation}
while on $T$, we get
\begin{equation}
    (\gamma^0-\gamma^3)\eta=0,
\end{equation}
i.e. we see that the boundary conditions are self-adjoint.
Under these self-adjoint boundary conditions, the condition that $(\eta,\eta)$ is minimised by $\epsilon$ is then equivalent to saying that equation (\ref{eq:variations}) vanishes for any variation $\delta \epsilon$ satisfying the appropriate boundary conditions, which implies
\begin{equation}\label{eq:adjoint problem}
    \hat{\mathcal{D}}^\dagger\eta = 0,
\end{equation}
with the aforementioned boundary conditions on $\eta$.

If, then, the kernel of $\hat{\mathcal{D}}^\dagger$ is trivial, this would imply that $\epsilon$ satisfies (\ref{eq:sw existence stuff}).
In the case of \cite{douganquasi}, the unmodified Sen-Witten operator is formally self-adjoint, and hence $\eta\equiv0$ as the kernel of the Sen-Witten operator with homogenous boundary conditions is trivial.
However, in the case under consideration the generalised Sen-Witten operator is not self-adjoint.
Similarly to \cite{reall2025lawblackholemechanics, chrusciel2003boundaryvalueproblemsdiractype}, we have the following supercovariant generalisation of the Sen-Witten operator on $\Sigma$:
\begin{equation}
    \hat{\mathcal{D}} \coloneqq \gamma^i D_i - \frac{1}{2}H \gamma^0 + \frac{1}{2}E_i\gamma^i \gamma^0 - \frac{1}{4}\epsilon_{ijk}B^i \gamma^{jk} + i\frac{3K}{\sqrt{2}},
\end{equation}
where $D_i$ is the induced $U(1)$-gauge covariant derivative on $\Sigma$, $H$ is the trace of the extrinsic curvature of $\Sigma$ in $M$ and $E_i$ and $B_j$ are the components of the electric and magnetic fields on $\Sigma$. It is clear that, upon taking the adjoint of this equation, the final term involving $K$ is not preserved (nor is the term with the magnetic field), and so this operator is not formally self-adjoint.

This issue can be circumvented by an argument given in \cite{reall2025lawblackholemechanics}: we write $\hat{\mathcal{D}}=\mathcal{D}+\Omega$ and view this as a map between appropriate Hilbert spaces, where $\mathcal{D}$ is the $U(1)$-gauged Sen-Witten operator and $\Omega$ contains the terms in $\hat{\cal D}$ involving the extrinsic curvature, Maxwell field and the cosmological constant.
In the case of interest, $\Sigma$ is compact with compact boundaries and hence $\Omega$ can be viewed as a compact operator between these spaces.
By a standard result of Fredholm theory, the indices (the difference between the dimensions of the kernel and cokernel) of $\hat{\mathcal{D}}$ and $\mathcal{D}$ are then equal.
Since $\cal D$ is self-adjoint, its index vanishes and hence so too does the index of $\hat{\cal D}$.
The uniqueness argument given above shows that $\text{dim}(\text{ker}(\hat{\mathcal{D}}))=0$ and so $\text{dim}(\text{coker}(\hat{\mathcal{D}}))=0$ (which is necessary and sufficient for surjectivity of $\hat{\mathcal{D}}$), thus $\hat{\mathcal{D}}$ is a bijection between these Hilbert spaces.
Bijectivity of $\hat{\mathcal{D}}$ establishes the existence and uniqueness of solutions to equation (\ref{eq:sw existence stuff}) with boundary conditions (\ref{eq:boundary conditions in dirac spinor notation}).

\section{The form of the stress-tensor in a supersymmetric spacetime} 

\label{appendix:form of stress tensor with null killing spinor}

In the proof of Theorem \ref{thm:equiv to 2.8}, we make use of a result that reads: given a spacetime $(M,g)$, at any point $p\in M$ such that a neighbourhood of $p$ admits a non-vanishing Killing spinor, we have
\begin{equation} \label{eq:result for theorem 2.8 introduction}
    (T_{ab}^{(m)} - \frac{1}{2}g_{ab}T^{(m)}) X^b = \frac{1}{2} T^{(m)}_{ab}X^b,
\end{equation}
with $T^{(m)}_{ab}$ is the stress-tensor of matter in $(M,g)$ defined by the Einstein equation.
In this Appendix, we will prove this statement.
Additionally, and for completeness, we will derive the allowed form of the stress-tensor and currents in a supersymmetric spacetime. This problem has been considered previously in \cite{caldarelli2004supersymmetric} but no details of the analysis were provided. We will confirm the result stated in \cite{caldarelli2004supersymmetric} for the case where $X^a$ is timelike, and we will also determine the result when $X^a$ is null.

In the following, we choose to change conventions, adopting those of \cite{Caldarelli_2003}. This change is purely for convenience in using various identities as presented in that paper: it is of no significance for the final results.
In particular, the metric signature is now ``mostly plus'', the gamma matrices satisfy $\{ \gamma^a, \gamma^b \} = 2g^{ab}$, $\gamma^5 \coloneqq i \gamma^0\gamma^1\gamma^2\gamma^3$ and Dirac conjugation is defined by $\bar\epsilon\coloneqq i\epsilon^\dagger\gamma^0$.
Additionally to this, the change in metric convention means the sign of the cosmological constant term in Einstein's equation is now reversed, as is the sign of the Maxwell stress-energy tensor. Moreover, Maxwell's equation becomes
\begin{equation}\label{eq:mw in other signature}
    \nabla^a F_{ab} = - 4\pi J_b.
\end{equation}

Given a spinor, following \cite{Caldarelli_2003}, we can construct various bilinears from that spinor. The five independent spinor bilinears are:\footnote{
The $1$-form $A_a$ should not be confused with the electromagnetic potential. We will not need to refer explicitly to the latter in this Appendix.} two scalars $f$ and $g$, two one-forms $X_a$ and $A_a$ and a two-form $\Psi_{ab}$.
These are defined as such
\begin{equation}
    \begin{split}
        f &\coloneqq \bar{\epsilon}\epsilon,\\
        g &\coloneqq i\bar{\epsilon}\gamma^5\epsilon,\\
        X_a &\coloneqq i\bar{\epsilon}\gamma_a\epsilon,\\
        A_a &\coloneqq i\bar{\epsilon}\gamma^5 \gamma_a\epsilon,\\
        \Psi_{ab} &\coloneqq i\bar{\epsilon}\gamma_{ab}\epsilon.
    \end{split}
\end{equation}
Factors of $i$ are necessary to ensure the reality of each form.
Among various identities, which we will introduce as needed, the two one-forms satisfy the following \cite{Caldarelli_2003}:
\begin{equation}
    X^2 = -A^2 = - (f^2 + g^2) \text{   and    } X\cdot A=0.
\end{equation}
$X^a$ cannot be zero (for $\epsilon$ non-vanishing) as in any orthonormal frame $X_0 = -\epsilon^\dagger\epsilon \neq 0$, and so it must be causal and non-zero.
From this, $A^a$ can be spacelike, null and proportional to $X^a$ or zero.

From the existence of a supercovariantly constant spinor, $\epsilon$, in a neighbourhood of a point $p\in M$, one obtains the following integrability condition in that neighbourhood \cite{Caldarelli_2003}
\begin{equation}
\begin{split}
     [\hat{\nabla}_b,\hat{\nabla}_a]\epsilon = &\Big\{ \frac{1}{l}(\star F_{ba}\gamma^5 - iF_{ba}) +\frac{1}{2l^2}\gamma_{ba} + \frac{1}{4}R^{cd}{_{ab}}\gamma_{cd} - F^{cd}F_{d [b}\gamma_{a]c} \\ & + \frac{1}{4}F^2 \gamma_{ba} - \frac{i}{l}F^{c}{_{[b}} \gamma_{a]c} - \frac{i}{2}\gamma_{cd[b}\nabla_{a]}F^{cd} -i \nabla_{[b}F{_{a]}}{^c}\gamma_c\Big\}\epsilon = 0,
\end{split}
\end{equation}
where $1/l^2 = -3\Lambda = K^2/2$, $F^2\coloneqq F_{ab}F^{ab}$ and $\gamma_{ab}\coloneqq \gamma_{[a}\gamma_{b]}$.

Contracting this with $\gamma^a$
we obtain, after some algebra,
\begin{equation}\label{eq:integrability}
    \Big\{E_{ab}\gamma^b  + i\nabla_c F^{cb}\gamma_{b}\gamma_{a} + \frac{1}{2}\big[\nabla^{[b}F^{cd]}\epsilon_{bcda}\gamma^5 + 3i\nabla_{[a}F_{bd]}\gamma^{bd} \big]\Big\}\epsilon = 0,
\end{equation}
where we have introduced the tensor
\begin{equation}
    E_{ab} \coloneqq R_{ab} - \Lambda g_{ab} + 2 F_{ac}F^c{_b} + \frac{1}{2}g_{ab}F^2.
\end{equation}
The energy-momentum tensor of the matter is defined via the Einstein equation:
\begin{equation}\label{eq:trace reversed einstein}
    E_{ab} = 8\pi(T_{ab}^{(m)} - \frac{1}{2}g_{ab}T^{(m)}),
\end{equation}

We now impose Maxwell's equations (\ref{eq:mw in other signature}). Equation (\ref{eq:integrability}) then becomes
\begin{equation}\label{eq:integrability with maxwell}
    \Big\{E_{ab}\gamma^b  - 4\pi i J^b\gamma_{b}\gamma_{a} \Big\}\epsilon = 0.
\end{equation}
We now show that $J_a$ must be causal (or zero).
To do so, contract this equation with $\gamma^a$ to obtain
\begin{equation} \label{eq:trace of stress-energy relation to current}
    \Big\{T{^{(m)}}- i J^a \gamma_a  \Big\}\epsilon = 0.
\end{equation}
Then, acting with $T{^{(m)}}+i J^b\gamma_b$, one obtains
\begin{equation}
    \Big\{T{^{(m)}}^2 + J\cdot J\Big\}\epsilon=0.
\end{equation}
Since $\epsilon$ is not zero at $p$, we see that, for $T$ to be real we must have $J\cdot J\leq0$ and hence $J$ is causal (or zero).

Contracting equation \eqref{eq:integrability with maxwell} with $i\bar{\epsilon}$, we obtain
\begin{equation} \label{eq:real and imaginary parts of the integrability condition}
    (E_{ab}X^b + 4\pi f J_a) - 4\pi iJ^b\Psi_{ba} = 0,    
\end{equation}
where we have used the definitions of the various relevant spinor bilinears.
The spinor bilinears are real so taking the imaginary part gives gives
\begin{equation}\label{eq:imaginary part of integrability}
    i_J \Psi = 0
\end{equation}
where $i_J$ denotes contraction of the vector $J^a$ with the first index of a differential form.
Using this and contracting \eqref{eq:trace of stress-energy relation to current} with $i\bar\epsilon\gamma_b$ we get
\begin{equation} \label{eq:trace times vector in terms of current}
    T{^{(m)}} X_b + f J_b = 0,
\end{equation}
from which we see that if $f=0$ then $T{^{(m)}}=0$ (as $X^a \neq 0$) and if $f\neq0$ then $J_a = -T{^{(m)}}f^{-1} X_a$.
We now split our analysis of the allowed form of the electric current into the cases when $X^a$ is timelike and null.

Firstly, the null case.
By contracting equation \eqref{eq:trace of stress-energy relation to current} with $\bar \epsilon$, we obtain that $J\cdot X = 0$ and hence $J$ is spacelike, zero or null and proportional to $X^a$.
But, $J$ must be causal or zero and hence
\begin{equation}
    J = \alpha X,
\end{equation}
where $\alpha$ is a smooth function which can be zero. 
By the requirement from Maxwell's equations that $\star J$ be closed, $\alpha$ must be constant along integral curves of $X^a$.
It remains to check that $i_J \Psi = 0$ is satisfied.
This follows immediately from equation (2.15) of \cite{Caldarelli_2003} which says that $i_X \Psi = 0$.
Note that this result is independent of the Einstein equation in the sense that we only needed to assume that $E_{ab}$ was symmetric and its trace real for us to obtain this result.

Secondly, the timelike case.
Contracting \eqref{eq:trace times vector in terms of current} with $T{^{(m)}}X^b - fJ^a$ we obtain
\begin{equation}
    T{^{(m)}}^2 X^2 = f^2 J^2.
\end{equation}
Using that $T{^{(m)}}^2 = -J^2$ we obtain
\begin{equation}
    (f^2 + g^2) J^2 = f^2 J^2,
\end{equation}
and so for $J^2 \neq 0$ we have $g=0$ which implies $f\neq 0$ as $X^a$ is timelike.
Then, equation \eqref{eq:trace times vector in terms of current} implies that $J_a = -T{^{(m)}}f^{-1} X_a$.
Defining $\alpha \coloneqq -T{^{(m)}}f^{-1}$ this gives that $J=\alpha X$ where $\alpha$ is non-zero.
As $g=0$, it again follows from equation (2.15) of \cite{Caldarelli_2003} that $i_J \Psi = 0$ is satisfied by this form of the current.
If $J^2 = 0$ we have that $T{^{(m)}} = 0$ and so contracting \eqref{eq:trace of stress-energy relation to current} with $\bar\epsilon$ we obtain that $J\cdot X = 0$ and so $J=0$ as $X$ is timelike (in particular, $J$ is not null).
Therefore, we allow $\alpha$ to be zero corresponding to the case when $J=0$ - which trivially satisfies $i_J \Psi = 0$ - and so
\begin{equation}
    J = \alpha X.
\end{equation}
Again, we require $\alpha$ to be constant along the integral curves of $X^a$ as Maxwell's equations demand that $\star J$ be closed.
As in the null case, this result is independent of the Einstein equation in the sense described above.

We now present the results necessary for the proof of Theorem \ref{thm:equiv to 2.8}.
We treat the timelike and null cases concurrently, noting that the latter just requires setting $f=0$ everywhere it appears in the sequel.
At the point $p$, we have $J= \alpha X$.
Thus, the condition \eqref{eq:integrability with maxwell} becomes
\begin{equation}\label{eq:integrability with bps current put in}
    \Big\{ E_{ab}\gamma^b - 4\pi i (\alpha X^b)\gamma_b \gamma_a \Big\} \epsilon = 0.
\end{equation}
Contracting with $\gamma^a$, we get
\begin{equation}
    T^{(m)} \epsilon = -\alpha f \epsilon,
\end{equation}
Since $\epsilon \neq 0$ at $p$, this implies $T{^{(m)}} = -\alpha f$.
Using this and contracting \eqref{eq:integrability with bps current put in} with $i\bar{\epsilon}$, we obtain
\begin{equation} \label{eq:result for theorem 2.8}
    (T{^{(m)}}_{ab} - \frac{1}{2}g_{ab}T{^{(m)}}) X^b = \frac{1}{2} T{^{(m)}}_{ab}X^b = -\frac{1}{2} \alpha f X_a,
\end{equation}
with $\alpha g=0$.
The first two equalities provide the results in \eqref{eq:result for theorem 2.8 introduction}.

For completeness, we now go on to derive the form of the stress-energy tensor at $p$ in the timelike and null cases.
We will see, in both cases, that the result is dust with velocity vector field $X^a$.
Additionally, this will allow us to further constrain $\alpha$ in both cases.
In both cases, we will need the following simplification of \eqref{eq:integrability with bps current put in}
\begin{equation}\label{eq:simplification of 87}
    (T{^{(m)}}_{ab}\gamma^b - i\alpha X_a)\epsilon = 0,
\end{equation}
obtained using that $X^a\gamma_a \epsilon = if\epsilon$ and $T^{(m)}=-\alpha f$.

In the null case, we choose a basis\footnote{
In this Appendix, Latin indices refer to an arbitrary orthonormal basis and Greek indices to a specific orthonormal basis.} in which $X^\mu = (1,1,0,0)$.
The condition that $T{^{(m)}}_{ab}X^b = 0$ implies that 
\begin{equation}\label{eq:stress-tensor in o/n basis null case conditions}
    T{^{(m)}}_{0A} + T{^{(m)}}_{1A} = 0, \qquad T{^{(m)}}_{00} + T{^{(m)}}_{01} = 0 = T{^{(m)}}_{11} + T{^{(m)}}_{10}
\end{equation}
where $A=2,3$.
We firstly consider expression \eqref{eq:simplification of 87} in this basis:
\begin{equation}
    (T{^{(m)}}_{\mu\nu}\gamma^\nu - i\alpha X_\mu)\epsilon = 0.
\end{equation}
Temporarily suspending the summation convention and multiplying the above with $T{^{(m)}}_{\mu\sigma}\gamma^\sigma + i\alpha X_\mu$ (where $\mu$ is some fixed index), we have
\begin{equation}\label{eq:integrability in o/n basis no summation convention}
    (T{^{(m)}}_{\mu\nu}T{^{(m)}}{_\mu}{^\nu} + \alpha^2 X_\mu X_\mu)\epsilon = 0.
\end{equation}
At $p$, $\epsilon \neq 0$, and so we get that $T{^{(m)}}_{\mu\nu}T{^{(m)}}{_\mu}{^\nu} + \alpha^2 X_\mu X_\mu =0$.
By considering $\mu = 0$, we obtain, after using \eqref{eq:stress-tensor in o/n basis null case conditions}, that $\Sigma_A T{^{(m)}}_{0A}T_{0A} + \alpha^2= 0$ and since this is a sum of non-negative terms, we must have that both are zero, i.e. $T{^{(m)}}_{0A} = 0$ and $\alpha = 0$.
By virtue of \eqref{eq:stress-tensor in o/n basis null case conditions}, this also implies $T{^{(m)}}_{1A} = 0$.
Doing similarly for $\mu =2,3$ we get that $T{^{(m)}}_{AB}$ = 0.
Therefore, the only non-zero components of $T{^{(m)}}_{\mu\nu}$ are $T{^{(m)}}_{00},T{^{(m)}}_{01}$ and $T{^{(m)}}_{11}$, only one of which is independent by virtue of \eqref{eq:stress-tensor in o/n basis null case conditions}.
Defining $\chi \coloneq T{^{(m)}}_{00}$, we have that $T_{01}= -\chi = - T_{11}$ and so $T{^{(m)}}_{\mu\nu} = \chi X_{\mu}X_{\nu}$. 
Since this is a tensor equation
\begin{equation}
\label{eq:dust}
    T{^{(m)}}_{ab} = \chi X_a X_b.
\end{equation}
A consequence of the above is that $\alpha=0$ and so the electric current of this matter must be zero. This implies that $T{^{(m)}}_{ab}$ must be conserved so $\chi$ must be constant along integral curves of $X$.
So, in the null case, matter is constrained to be uncharged, null dust, as in the ungauged case \cite{tod1983all}.

In the timelike case, we work in an orthonormal basis adapted to $X^a$.
Much of the computation is similar: the idea is the same and so we will not explain the steps in such detail.
We now choose an orthonormal basis in which $X^\mu = |f| (1,0,0,0)$ and again contract \eqref{eq:simplification of 87} with $T{^{(m)}}_{\mu\sigma}\gamma^\sigma + i\alpha X_\mu$ to obtain \eqref{eq:integrability in o/n basis no summation convention}, suspending the summation convention on the fixed index $\mu$.
Looking at the $\mu=1,2,3$ components and using that $T_{0i}= 0$ by virtue of \eqref{eq:result for theorem 2.8}, we get that $T_{ij} = 0$ and so we see $T_{00}$ is the only non-zero component of the stress-tensor.
Hence, $T{^{(m)}}_{\mu \nu} = T_{00}f^{-2} X_\mu X_\nu$ at $p$ in this frame.
But, since $T{^{(m)}}=-T{^{(m)}}_{00}=-\alpha f$ we get that $T{^{(m)}}_{\mu\nu} = \alpha f^{-1} X_\mu X_\nu$.
Since this is a tensor equation, it is valid in any basis and hence \eqref{eq:dust} holds 
at $p$, where $\chi\coloneq \alpha f^{-1}$ is a smooth function constant along the flow of $X^a$ (by virtue of both $\alpha$ and $f$ being constant along the flow of $X^a$).
Using this definition of $\chi$, we have that 
\begin{equation}
   \label{eq:current}
    J_a = \chi fX_a.
\end{equation}
So, we find that, in the timelike case, matter is constrained to be electrically charged, timelike dust, in agreement with \cite{caldarelli2004supersymmetric}.

We can summarize the results of this analysis without splitting into timelike or null cases by saying that equations \eqref{eq:dust} and \eqref{eq:current} hold in general, along with $\chi g=0$.


\begin{thebibliography}{99}
\bibitem{bardeen1973four}
J.~M.~Bardeen, B.~Carter and S.~W.~Hawking,
``The Four laws of black hole mechanics,''
Commun. Math. Phys. \textbf{31}, 161-170 (1973)
doi:10.1007/BF01645742

\bibitem{israel1986third}
W.~Israel,
``Third Law of Black-Hole Dynamics: A Formulation and Proof,''
Phys. Rev. Lett. \textbf{57}, no.4, 397 (1986)
doi:10.1103/PhysRevLett.57.397

\bibitem{kehle2025gravitational}
C.~Kehle and R.~Unger,
``Gravitational collapse to extremal black holes and the third law of black hole thermodynamics,''
[arXiv:2211.15742 [gr-qc]].

\bibitem{kehle2024extremalblackholeformation}
C.~Kehle and R.~Unger,
``Extremal black hole formation as a critical phenomenon,''
[arXiv:2402.10190 [gr-qc]].

\bibitem{gibbons1982bogomolny}
G.~W.~Gibbons and C.~M.~Hull,
``A Bogomolny Bound for General Relativity and Solitons in N=2 Supergravity,''
Phys. Lett. B \textbf{109}, 190-194 (1982)
doi:10.1016/0370-2693(82)90751-1

\bibitem{ghhp}
G.~W.~Gibbons, S.~W.~Hawking, G.~T.~Horowitz and M.~J.~Perry,
``Positive Mass Theorems for Black Holes,''
Commun. Math. Phys. \textbf{88}, 295 (1983)
doi:10.1007/BF01213209

\bibitem{reall2025lawblackholemechanics}
H.~S.~Reall,
``A third law of black hole mechanics for supersymmetric black holes and a quasi-local mass-charge inequality," Phys.Rev. D {\bf 110} 124059 (2025)
[arXiv:gr-qc/2410.11956 [gr-qc]].

\bibitem{douganquasi}
A.~J.~Dougan and L.~J.~Mason,
``Quasilocal mass constructions with positive energy,''
Phys. Rev. Lett. \textbf{67}, 2119-2122 (1991)
doi:10.1103/PhysRevLett.67.2119

\bibitem{rallabhandi2025spinorialquasilocalmassspacetimes}
V.~Rallabhandi,
``Spinorial quasilocal mass for spacetimes with negative cosmological constant,"
[arXiv:gr-qc/2504.11971 [grqc]].

\bibitem{nester1981new}
J.~A.~Nester,
``A New gravitational energy expression with a simple positivity proof,''
Phys. Lett. A \textbf{83}, 241 (1981)
doi:10.1016/0375-9601(81)90972-5

\bibitem{Kosteleck__1996}
A.~V.~Kostelecký and M.~J.~Perry,
``Solitonic black holes in gauged N = 2 supergravity,"
Phys. Lett. B \textbf{371} (1996), 191-198
[arXiv:hep-th/9512222 [hep-th]].

\bibitem{caldarelli1999supersymmetry}
M.~M.~Caldarelli and D.~Klemm,
``Supersymmetry of anti-de Sitter black holes,"
Nucl. Phys. B \textbf{545} (1999), 434-460
doi:10.1016/S0550-3213(98)00846-3
[arXiv:hep-th/9808097 [hep-th]].

\bibitem{penrose1984spinorsvol1}
R.~Penrose and W.~Rindler,
``Spinors and space-time,'' Vol 1, Cambridge Univ. Press, 1988,
doi:10.1017/CBO9780511564048

\bibitem{geroch1973space}
R.~P.~Geroch, A.~Held and R.~Penrose,
``A space-time calculus based on pairs of null directions,''
J. Math. Phys. \textbf{14}, 874-881 (1973)
doi:10.1063/1.1666410

\bibitem{FREEDMAN1977221}
D.~Z.~Freedman and A.~Das,
``Gauge internal symmetry in extended supergravity,"
Nucl. Phys. B \textbf{120} (1977), 221-230
doi:10.1016/0550-3213(77)90041-4

\bibitem{Romans_1992}
L.~J.~Romans,
``Supersymmetric, cold and lukewarm black holes in cosmological Einstein-Maxwell theory,"
Nucl. Phys. B \textbf{383} (1992), 395-415
[arXiv:hep-th/9203018 [hep-th]].

\bibitem{tod1983all}
K.~P.~Tod,
``All Metrics Admitting Supercovariantly Constant Spinors,''
Phys. Lett. B \textbf{121}, 241-244 (1983)
doi:10.1016/0370-2693(83)90797-9

\bibitem{Rogatko_2002}
M.~Rogatko,
``Positivity of energy in Einstein-Maxwell axion dilaton gravity,''
Class. Quant. Grav. \textbf{19}, 5063-5072 (2002)
doi:10.1088/0264-9381/19/20/303
[arXiv:hep-th/0209126 [hep-th]].

\bibitem{witten1981new}
E.~Witten,
``A Simple Proof of the Positive Energy Theorem,''
Commun. Math. Phys. \textbf{80}, 381 (1981)
doi:10.1007/BF01208277

\bibitem{Caldarelli_2003}
M.~M.~Caldarelli and D.~Klemm,
``All supersymmetric solutions of N=2,D=4 gauged supergravity,"
JHEP \textbf{09} (2003), 019
doi:10.1088/1126-6708/2003/09/019
[arXiv:hep-th/0307022 [hep-th]].

\bibitem{Beig_1996}
R.~Beig and P.~T.~Chrusciel,
``Killing vectors in asymptotically flat space-times: I. Asymptotically translational Killing vectors and the rigid positive energy theorem,"
J. Math. Phys. \textbf{37} (1996), 1939-1961
doi:10.1063/1.531497
[arXiv:gr-qc/9510015 [gr-qc]].

\bibitem{Chru_ciel_2006}
P.~T.~Chrusciel, H.~S.~Reall and P.~Tod,
``On Israel-Wilson-Perjes black holes,''
Class. Quant. Grav. \textbf{23}, 2519-2540 (2006)
doi:10.1088/0264-9381/23/7/018
[arXiv:gr-qc/0512116 [gr-qc]].

\bibitem{caldarelli2004supersymmetric}
M.~M.~Caldarelli and D.~Klemm,
``Supersymmetric Gödel-type universe in four dimensions,"
Class. Quant. Grav. \textbf{21} (2004), L17-L20
doi:10.1088/0264-9381/21/4/L03
[arXiv:hep-th/0310081 [hep-th]].

\bibitem{chrusciel2003boundaryvalueproblemsdiractype}
R.~Bartnik and P.~T.~Chrusciel,
``Boundary value problems for Dirac type equations, with applications,''
[arXiv:math/0307278 [math.DG]].

\bibitem{Kehle:2023eni}
C.~Kehle and R.~Unger,
``Event horizon gluing and black hole formation in vacuum: The very slowly rotating case,''
Adv. Math. \textbf{452}, 109816 (2024)
doi:10.1016/j.aim.2024.109816
[arXiv:2304.08455 [gr-qc]].

\bibitem{Murata:2013daa}
K.~Murata, H.~S.~Reall and N.~Tanahashi,
``What happens at the horizon(s) of an extreme black hole?,''
Class. Quant. Grav. \textbf{30}, 235007 (2013)
doi:10.1088/0264-9381/30/23/235007
[arXiv:1307.6800 [gr-qc]].

\bibitem{Angelopoulos:2024yev}
Y.~Angelopoulos, C.~Kehle and R.~Unger,
``Nonlinear stability of extremal Reissner-Nordstr{\"o}m black holes in spherical symmetry,''
[arXiv:2410.16234 [gr-qc]].

\end{thebibliography}
\end{document}